\documentclass[journal,onecolumn]{IEEEtran}
\usepackage{amsmath}

\usepackage{amssymb}
\usepackage{amsthm}
\usepackage{dsfont}
\usepackage{url}
\usepackage{enumerate}

\usepackage{graphicx,color,ifpdf,subfigure}

\graphicspath{{Fig/}{.}}
\ifpdf
\DeclareGraphicsExtensions{.pdf}
\else
\DeclareGraphicsExtensions{.eps}
\fi

\def\<{\langle}
\def\>{\rangle}

\def\bx{\boldsymbol{x}}

\def\bX{\boldsymbol{X}}

\def\by{\boldsymbol{y}}

\def\bu{\boldsymbol{u}}

\def\bz{\boldsymbol{z}}

\def\bmu{\boldsymbol{\mu}}
\def\bxi{\boldsymbol{\xi}}

\def\b1{\boldsymbol{1}}

\newtheorem{thm}{Theorem}[section]
\newtheorem{prop}[thm]{Proposition}
\newtheorem{lemma}{Lemma}
\newtheorem{corollary}[thm]{Corollary}

\begin{document}
% Title.
% ------
\title{Isotropic Multiple Scattering Processes on Hyperspheres}
\author{Nicolas~Le~Bihan, Florent~Chatelain,
and~Jonathan~H.~Manton,~\IEEEmembership{IEEE~fellow}
\IEEEcompsocitemizethanks{
\IEEEcompsocthanksitem N. Le Bihan is with the CNRS, University of Melbourne, Australia. email: {\tt nicolas.le-bihan@gipsa-lab.grenoble-inp.fr}. His research was supported by the ERA, European Union, through the International Outgoing Fellowship (IOF GeoSToSip 326176) program of the 7th PCRD.
\protect\\
\IEEEcompsocthanksitem F. Chatelain is with the GIPSA-Lab, Department of Images and Signal, Grenoble, France. email: {\tt florent.chatelain@gipsa-lab.grenoble-inp.fr}.\protect\\
\IEEEcompsocthanksitem J.H. Manton is with the University of Melbourne, Australia. email: {\tt jmanton@unimelb.edu.au} }
}

\maketitle
\begin{abstract}

This paper presents several results about isotropic random walks and multiple scattering 
processes on hyperspheres 
${\mathbb S}^{p-1}$.
It allows one to derive the Fourier expansions on ${\mathbb S}^{p-1}$ of these processes.
A result of unimodality for the multiconvolution of symmetrical 
probability density functions (pdf) on ${\mathbb S}^{p-1}$ is also introduced.
Such processes are then studied in the case where the scattering distribution is 
von Mises Fisher (vMF). Asymptotic distributions for the multiconvolution of 
vMFs on ${\mathbb S}^{p-1}$ are obtained.  Both Fourier expansion and
asymptotic approximation allows us to compute estimation bounds for 
the parameters of Compound Cox Processes (CCP) on ${\mathbb S}^{p-1}$. 

\end{abstract}
\section{Introduction}
\label{sec:intro}
Mixtures of von Mises-Fisher (vMF) distributions are models used in applications
ranging from MRI data analysis \cite{Bhalerao} to radiation therapy beam direction
clustering \cite{Bangert} and speaker clustering \cite{Tang}. The finite
mixture case was originally studied in \cite{Banerjee} for data clustering on
hyperspheres. All the above mentioned contributions made use of EM algorithms for the estimation of
mixtures weights, vMF distribution parameters or number of
mixture component. 

In this paper, we consider particular countably infinite mixtures of
directional distributions where the components of the mixture 
are multiply convolved distribution of unit vector in $\mathbb{R}^{p}$, {\em i.e.} elements of ${\mathbb S}^{p-1}$, 
and where the weights are controlled by a Cox process.

The proposed approach is valid for any dimension $p$, 
even though applications are mainly concerned with the case $p=3$. In particular, the problem of multiple
scattering for waves (or particles) in a random medium can be studied with the random processes presented
in this work. As originally introduced in \cite{Ning}, compound/mixture processes model allow the description 
of the output distribution of the direction
of propagation of the wave in terms of a mixture of symmetrical distributions on hyperspheres. 
Similar models are studied in \cite{LeBihanMargerin,Said} where multiple scattering is described as a 
Compound Poisson Process (CPP) on the rotation group $SO(3)$. In \cite{LeBihanMargerin}, it is shown 
that this model allows to describe forward multiple scattering, and its accuracy is high when the number
of  diffusion events is low. Thus, the CPP model describes the behavior
of particles in a scattering medium before the fully developed diffusive regime
(known to be thoroughly described by the Brownian motion on $SO(3)$ and
originally studied by Perrin \cite{Perrin}). 

Convolution on the hypersphere have been studied for {\em pdf} modelling  in engineering application \cite{Dokmanic2010}.
Random walk models on the sphere have been used to describe wave propagation in engineering litterature \cite{Franceschetti2004} and multiple scattering is a reccurent issue in many engineering applications such as optics \cite{Durant2007}, communications \cite{Jin2012} and antennas \cite{ghogho2001,Costa2014}. The occurence of random scatterers in a wide range of wavelengths makes the problem of multiple scattering relevent to many types of wave propagation. Being able to analyse the distribution of the multiply scattered wave/particle is thus of critical importance. The model studied in this paper aims at describing multiple scattering with a stochastic process and make use of results from harmonic analysis on spheres to predict the behaviour of multiple scattering processes. In addition, stochastic process model allows to infer on the medium the particle/wave has travelled through. In this paper, the harmonic expansions we have derived allows one to numerically evaluate the lower bounds achievable for the estimation of medium parameters.

In this paper, we extend the CPP model to the case where the counting process is no longer a homogeneous Poisson process, 
but rather a Cox process, {\em i.e.} a process with intensity being a positive random process itself. 
The family of processes considered are thus Compound Cox Processes (CCP) on ${\mathbb S}^{p-1}$. 
In contrast with \cite{LeBihanMargerin,Said}, the pdf of the random walk and the multiple scattering process
is here studied in detail for the general case of isotropic random scattering events
and when these events are von Mises Fisher distributed.
%here studied in details in the general isotropic case, and when it is a von Mises Fisher. 
Several results about  multiconvolution, symmetry and unimodality of such pdfs on ${\mathbb S}^{p-1}$ 
are introduced and used to provide Fourier series expansion of the pdf of a multiply scattered unit vector in
$\mathbb{R}^{p}$. In the von Mises Fisher case, we provide an asymptotic distribution of the process which is 
a mixture of vMF distributions. In addition, we compute the Cramer Rao lower bounds (CRLB) for some 
parameters of the CCP model on ${\mathbb S}^{p-1}$ in the case where the counting process is 
either a Poisson process or a Cox process with distribution belonging to an exponential family
(when its intensity process is a Gamma process).

The contributions of the paper can be summarized as follows: the isotropic multiple scattering process is expressed using multiple convolution on double cosets. An unimodality theorem is given for such multi-convolution of unimodal and symmetric {\em pdf}s on ${\mathbb S}^{p-1}$.
Using harmonic analysis results on hyperspheres allows us
to obtain the expression of the Fourier coefficients (Legendre polynomial moments) of the {\em pdf}
of a $n$-step isotropic random walk on ${\mathbb S}^{p-1}$, which leads to the Fourier expansion of this {\em pdf}. 
These results are extended to the multiple scattering process
{\em pdf} when the occurence of scattering events is a Cox process. In particular, the case when each random step follows a von Mises Fisher (vMF) law is studied in detail: asymptotic approximations for vMF random walk and multiple scattering process on ${\mathbb S}^{p-1}$ are given.

The remainder of this text is outlined as follows. 
Properties of the isotropic random walk on ${\mathbb S}^{p-1}$ are given in \ref{sec:RWSP}. 
Section \ref{sec:mulscatmod} presents the Compound Cox Process (CCP) model for the study 
of multiple scattering on hyperspheres, with emphasis on the use of harmonic analysis 
on ${\mathbb S}^{p-1}$ to provide characteristic function of the distribution after a time $t$. 
Section \ref{sec:VMFRWSP} gives an approximation result for multiconvolved vMF pdfs and 
its potential use for the estimation of the CCP parameters.

%%%%%%%%%%%%%%%%%%%%%%%%%%%%
%
%  Random Walk S^{n-1}
%
%%%%%%%%%%%%%%%%%%%%%%%%%%%%
\section{Random walk on ${\mathbb S}^{p-1}$}
\label{sec:RWSP}

After reviewing some known facts about hyperspheres and functions taking values on hyperspheres, 
we introduce new results for the homogeneous random walk on ${\mathbb S}^{p-1}$.

\subsection{General properties}
\label{subsec:Gene_prop}
In ${\mathbb R}^p$, the hypersphere, denoted ${\mathbb S}^{p-1}$, is  the set of $p$
dimensional vectors with unit length ${\mathbb S}^{p-1}=\{ \bx \in \mathbb{R}^{p} ; \, || \bx || = 1 \}$.
Hyperspheres (sometimes simply called spheres) are well-known compact manifolds with positive curvature.
They are homogeneous spaces of importance in Lie group theory, especially because of their relation with the rotation group $SO(p)$;
as they are the following quotients: ${\mathbb S}^{p-1}\cong SO(p)/SO(p-1)$. Hyperspheres are Riemannian symmetric spaces for
which the Riemannian distance is simply the "angle" between two elements, $d({\bf x},{\bf y})=|\arccos(\bx^T\by)|$ for $\bx,\by \in {\mathbb S}^{p-1}$.
In the sequel, we will make use of the notation $\theta_{\bx,\by}$ for the distance between $\bx$ and $\by$, which is comprised between
$0$ and $\pi$, to avoid any ambiguity.
Also, we will use the notation $\bmu^T \bx = \cos\theta_{\bx,\bmu}$ as $\bx$ and $\bmu$ are unit vectors in ${\mathbb R}^p$. Finally, 
this allows us to define the tangent-normal decomposition of the random unit vector $\bx$:
\begin{align}
\label{eq:tangentnorm}
 \bx &= t \bmu + \sqrt{1-t^2} \bxi,
\end{align}
where $t=\bmu^T \bx$, and $\bxi$ belongs to the intersection of ${\mathbb S}^{p-1}$ with
the hyperplane through the origin normal to $\bmu$, denoted as $\bmu^{\perp} \cap {\mathbb S}^{p-1}$,
which equals ${\mathbb S}^{p-2}$.

\subsection{Mathematical problem statement}
\label{sec:math}
We consider the problem of modelling the distribution of the isotropic multiple scattering process on the sphere ${\mathbb S}^{p-1}$ in ${\mathbb R}^p$.
%after a given number $N$ of scattering events. 
Each isotropic scattering event acts as a random rotation $R \in SO(p)$ on the direction of propagation $\bx \in {\mathbb S}^{p-1}$. 
The direction after $k\ge 1$ scattering events reads
\begin{align*}
 \bx_k &= R_k \bx_{k-1}= R_k \ldots R_1 \bx_0,
\end{align*}
where $R_l$ is the rotation matrix associated with the $l$th scattering event.
Let $\bxi_k$ be the direction of the normal part of $\bx_k$ with respect to $\bx_{k-1}$ as defined in (1). 
The isotropy assumption involves that $\bxi_k$  is uniformly distributed on 
$\bx_{k-1}^{\perp} \cap {\mathbb S}^{p-1}$.
Moreover, the scattering events are assumed to be mutually independent and independent of the initial direction
$\bx_0$. This involves the first order Markov property on the chain of random directions
$\bx_0$, $\bx_1$, $\ldots$ $\bx_n$ in ${\mathbb S}^{p-1}$:  $\bx_k$ given 
$\bx_{k-1}$ is independent of $\bx_l$ for $0 \le l \le k-2$.

The multiple scattering process considered in this paper is governed by the distribution of $\bx_t \equiv \bx_{N(t)} \in {\mathbb S}^{p-1}$ where the random number $N(t)$ of scatterers/rotations after a time $t$
is driven by a Poisson or more generally a Cox process.

\subsection{Symmetrical pdfs on ${\mathbb S}^{p-1}$}
\label{subsec:characf}
In this paper, we will consider probability density functions (\emph{pdf}s) $f$ of the direction vector $\bx$. These
\emph{pdf}s are elements from
$L^1({\mathbb S}^{p-1},{\mathbb R})$ with the following additional constraints: positivity and 
$\int_{{\mathbb S}^{p-1}}f({\bx})d{\bx}=1$.
In particular, we will be concerned with \emph{pdf}s that will only depend on the angular variable
(the Riemannian distance introduced in \ref{subsec:Gene_prop}).
An example of such \emph{pdf} is the von Mises Fisher \cite{Mardia} distribution that will be considered later in the paper.
Before moving to this specific case, we consider general symmetrical \emph{pdf}s. In this case,
the \emph{pdf} $f(\bx ; \bmu)$ is only a function of the cosine $\bx^T \bmu$, i.e. for all
$\bx \in {\mathbb S}^{p-1}$
\begin{align}
\label{eq:pdfsym}
f(\bx ; \bmu)= g(\bmu^T \bx),
\end{align}
where the unit vector parameter
$ \bmu \in {\mathbb S}^{p-1}$ characterizes the rotational axis \cite[p. 179]{Mardia}.
Without loss of generality, we can assume in the sequel that $\bmu$ is oriented such that $E[ \bmu^T \bx ] \ge 0$.
Consider the tangent-normal decomposition $\bx = t \bmu + \sqrt{1-t^2} \bxi$
defined in \eqref{eq:tangentnorm}. 
The rotational symmetry constraint \eqref{eq:pdfsym} directly implies that  $\bxi$ is independent
of $t$ and uniformly distributed on the space $\bmu^{\perp} \cap {\mathbb S}^{p-1}$.
For symmetry reason, the first moment of area of this space
is the null vector. Thus $E[\bxi]=\boldsymbol{0}$, and the mean of $\bx$ expresses as
\begin{align}
\label{eq:MomStep}
  E[ \bx ] &=  \rho \bmu,
\end{align}
where the scalar $\rho=E[\bmu^T \bx]= ||E[\bx]|| \in [0,1]$ is called the  {\it mean resultant
length} \cite[p. 164]{Mardia}.
When $\rho>0$, the {\it mean direction} is
uniquely defined as the rotational axis vector $ \bmu \in {\mathbb S}^{p-1}$.
Note that the mean resultant
length $\rho$ is directly linked with the dispersion of the directional distribution.
A value of $\rho$ close to $1$ indicates a high concentration about the mean direction.

Based now on the higher moments of directional statistics, harmonic analysis on spheres provides us a way to derive
a characteristic function for \emph{pdfs} taking values on ${\mathbb S}^{p-1}$. In the symmetrical case
(also known as the \emph{zonal} case), the characteristic function ({\em i.e.}
the Fourier transform of $f$) takes a simple form \cite{Kent78,Volker2013} as
an harmonic basis consists of the Legendre polynomials.
Given a \emph{pdf} $f \in L^1({\mathbb S}^{p-1},{\mathbb R})$ that is symmetrical about $\bmu \in {\mathbb S}^{p-1}$,
its characteristic function, denoted $\widehat{f}_{\ell}$, for $\ell \ge 0$, is given by:
\begin{align}
\widehat{f}_{\ell} & = E\left[ P_{\ell}(\cos\theta_{\bx,\bmu})\right], \\
& = E\left[ P_{\ell}(\bmu^T \bx )\right],\\
 & =\int_{{\mathbb S}^{p-1}} f( \bx ; \bmu ) P_{\ell}( \bmu^T \bx ) d\bx
\label{eq:charact_func_pdf_S}
\end{align}
where $P_{\ell}( \bmu^T \bx )$ are the Legendre polynomials of order $\ell$ in dimensions $p$
taken at $\bx$ with respect to $\bmu$,
the symmetry axis of $f$. 
Note that the Legendre polynomials $P_{\ell}(t )$ in dimensions $p$ are the same as the ultraspherical or 
Gegenbauer polynomials $C_{\ell}^{(p-2)/2}(t)$ \cite[pp. 771--802]{Abramowitz} renormalized such that $P_{\ell}(1)=1$. This yields
\begin{align*}
 P_{\ell}( t ) &=
 \Big[ \underbrace{\frac{\Gamma(\ell+p-2)}{\ell! \Gamma(p-2)}}_{C_{\ell}^{(p-2)/2}(1)} \Big]^{-1} C_{\ell}^{(p-2)/2}(t)  ,
\end{align*}
for all $t\in[-1,1]$, $\ell\ge 0$. This normalization ensures that $|P_{\ell}(t )|\le 1$ for all
$|t| \le 1$, thus $| \widehat{f}_{\ell}| \le 1$.
This family of polynomials
forms an orthogonal basis on the Hilbert space of square-integrable functions
on ${\mathbb S}^{p-1}$ that are rotationally symmetric about $\bmu$:
\begin{align*}
 \langle P_{\ell} , P_{m} \rangle &= \int_{{\mathbb S}^{p-1}} P_{\ell}(\bmu^T\bx)P_{m}(\bmu^T\bx) d\bx =
 c_{p,\ell}^{-1} \delta_{\ell,m},
\end{align*}
where $\delta_{\ell,m}=1$ if $\ell=m$, and $0$ otherwise. The normalizing constants read
\begin{align}
c_{p,\ell}
%&= \frac{1}{\omega_{q-1}}  \frac{ (2 \ell + p - 2 )  }{ p-2 } C_{\ell}^{(p-2)/2}(1),\nonumber \\
&= \frac{1}{\omega_{q-1}}
\frac{(2 \ell + p - 2 ) \Gamma(\ell + p -2 ) }{ \ell !  \Gamma(p-1) },
\label{eq:normalcst}
\end{align}
for all $\ell \ge 0$ with $\omega_{q-1}= 2 \frac{\pi^{p/2}}{\Gamma(p/2)}$ the area of the $(p-1)$-dimensional sphere
${\mathbb S}^{p-1}$.
Moreover, the Fourier expansion of any $p$-dimensional rotationally symmetric and continuous
\emph{pdf} $f(\cdot;\bmu)$ can be written:
\begin{align}
f(\bx ;\bmu) = \sum_{\ell \geq 0} c_{p,\ell} \widehat{f}_{\ell}P_{\ell}( \bmu^T \bx),
\label{eq:fourierexp}
\end{align}
for all $\bx \in {\mathbb S}^{p-1}$.
Note that the Fourier coefficients $\widehat{f}_{\ell}$, also called {\em Legendre polynomial moments},
are scalar valued and that is a consequence of the symmetry assumption.

\subsection{Convolution of symmetrical \emph{pdf}s on ${\mathbb S}^{p-1}$}
\label{sec:coset}
As already mentioned, we will focus on \emph{pdf}s which are symmetrical and thus only depend on one angular variable.
In particular, this will come from the fact that we will consider isotropic random walks on ${\mathbb S}^{p-1}$. As we will also assume the random steps to be independent, we will end up considering multiple convolution of their associated density. Following \cite{Dym1972}, we provide here a way to handle the convolution of symmetrical functions in the framework of integrable functions over double cosets.
First, recall that given a group $G$ and two subgroups of $G$, denoted $H$ and $K$, a double coset in $G$ is an equivalence class defined by the equivalence relation $x \sim y$ \emph{iff} there exists $h \in H$ and $k \in K$ such that:
\begin{equation}
hxk=y.
\end{equation}
Given $g\in G$, the double coset $H g K=\left\{ hgk \mid h \in H,\, k \in K\right\}$ is therefore the orbit of the
group action of $H \times K$ on $g$, where $H$ acts by left multiplication and $K$ acts by right multiplication.
The set of double cosets denoted as $H\backslash G/K$ contains all the orbits of the group action of $H \times K$ on $G$.
Here, we will consider the case where $G=SO(p)$ and $H=K=SO(p-1)$, \emph{i.e.} the double cosets in $SO(p-1)\backslash SO(p)/SO(p-1)$. As explained in \cite{Dym1972}, this set of double cosets can be parametrized
using the colatitude measured with respect to the axis left invariant by the $SO(p-1)$ subgroup.
The space $L^1(SO(p-1)\backslash SO(p)/SO(p-1), {\mathbb R})$ is the space of functions in $L^1(SO(p), {\mathbb R})$ that are
invariant on double cosets $KgK$ for $g \in SO(p)$ and $K\cong SO(p-1)$.

Thus, such functions can be thought as function $g(\bx^T \bmu)= g(\cos \theta_{\bx})$ of the (co)latitude of the
$(p-2)$-dimensional sphere defined by $\cos \theta_{\bx}= \bx^T \bmu$ where $\bmu \in {\mathbb S}^{p-1}$
is the axis left invariant by the $SO(p-1)$ rotation subgroup (to be chosen arbitrarily)%
\footnote{A function  $h \in L^1(SO(p-1)\backslash SO(p)/SO(p-1), {\mathbb R})$ is defined on the rotation 
group $SO(p)$. However, this function depends only on the cosine between the axis left invariant by the $SO(p-1)$ subgroup 
and its rotated image. For any $R \in SO(p)$, such a function can be expressed as  
$h(R)= f(\bx) = g(\bmu^T \bx)$ with $\bx= R \bmu$,
where the functions $f$ and $g$ are defined on $\mathbb{S}^{p-1}$ and  $[-1,1]$ respectively. 
This shows that 
a function in  $L^1(SO(p-1)\backslash SO(p)/SO(p-1), {\mathbb R})$ can be identified as a function on $\mathbb{S}^{p-1}$ or 
a function of the cosine $\bmu^T \bx$ defined on  $[-1,1]$. By abuse of notation, 
we will confuse these functions in the remainder.}.
This underlines that symmetrical {\em pdf}s about $\bmu$ belong to this double coset space according to eq. \eqref{eq:pdfsym}.

The convolution product in $L^1(SO(p-1)\backslash SO(p)/ SO(p-1), {\mathbb R})$ is inherited from the convolution
product in $L^1(SO(p),{\mathbb R})$ and reads:
\begin{align}
\label{eq:conv}
 \left( f \star_{\bmu} g \right) ({\bx}) & =
 \int_{{\mathbb S}^{p-1}} f({\bx}^T{\by})g({\by}^T\bmu) d\by,
\end{align}
for all $\bx \in {\mathbb S}^{p-1}$, where $\bmu$ is again the axis left invariant by the
rotation subgroup $SO(p-1)$.
The notation $\star_{\bmu}$ is used to recall that the (co)latitude is measured with respect to the axis $\bmu$
left invariant by $SO(p-1)$.

\begin{prop}[Convolution]
\label{prop:conv}
Let $f,g \in
 L^1(SO(p-1)\backslash SO(p)/ SO(p-1), {\mathbb R})$,
with $\bmu$ standing for the axis left invariant by the rotation subgroup $SO(p-1)$.
The following properties hold for the convolution in the double coset space.
\begin{enumerate}[i)]
 \item Stability.  $f \star_{\bmu} g \in
 L^1(SO(p-1)\backslash SO(p)/ SO(p-1), {\mathbb R})$, that is  $f \star_{\bmu} g$ is
 a function of the only cosine $\bmu^T {\bx}$:
 \begin{align}
 \label{eq:stab}
  \left( f \star_{\bmu} g \right) ({\bx}) & \equiv \left(f \star_{\bmu} g\right) (\bmu^T {\bx}),
 \end{align}
 \item Commutativity.
  \begin{align}
 \label{eq:commut}
  \left( f \star_{\bmu} g \right) ({\bx}) & = \left( g \star_{\bmu} f \right) ({\bx}),
 \end{align}
  \item Fourier product. If $\widehat{f}_{\ell}$ and $\widehat{g}_{\ell}$ are the respective  $\ell$th-order
  Fourier coefficients of $f$ and $g$,  the Fourier coefficients of their convolution product is given by:
\begin{align}
\label{eq:FourierConv}
\widehat{\left(f \star_{\bmu} g\right)}_{\ell}& =\widehat{f}_{\ell}\,\widehat{g}_{\ell},
\end{align}
for all $\ell \ge 0$.
\end{enumerate}
\end{prop}
\begin{proof}
Some proofs for the two first properties i) and ii) are given in  \cite[p. 237--239]{Dym1972} for $p=3$; they naturally extend
to the general case $p\ge 2$. Property iii) is a direct consequence of the Funk-Hecke theorem \cite[Theorem 7.8, p. 188]{Volker2013}.
\end{proof}
These properties illustrate that the convolution product and
its Fourier expansion behave nicely in this space. In particular, property iii) of Prop. \ref{prop:conv}
allows us to obtain the convolution theorem.
Finally, some well-known results on symmetry and unimodality in the real line can be extended
to the  convolution on the hypersphere.
\begin{thm}[Unimodality of the convolution product]
\label{thm:uni}
Let $f,g \in L^1(SO(p-1)\backslash SO(p)/ SO(p-1), {\mathbb R})$ be the {\it pdf}s of two absolutely continuous
unimodal and rotationally symmetric distributions on ${\mathbb S}^{p-1}$ with the same mode which
necessarily equals their mean direction $\bmu \in {\mathbb S}^{p-1}$.
Then the convolved distribution $f \star_{\bmu} g$, which is rotationally symmetric about $\bmu$ or equivalently belongs to
$L^1(SO(p-1)\backslash SO(p)/ SO(p-1), {\mathbb R})$
according to Prop. \ref{prop:conv}, is also unimodal with mode $\bmu$.
\end{thm}
\begin{proof} See Appendix \ref{app:uni}.
\end{proof}

\subsection{Random walk on ${\mathbb S}^{p-1}$: directional distribution}
\label{sec:RW}

As explained in \ref{sec:math}, the chain of random vectors
$\bx_0$, $\bx_1$, $\ldots$ $\bx_n$ in ${\mathbb S}^{p-1}$ obeys the Markov property
since given all the past directions $\bx_{0},\ldots,\bx_{k-1}$, the
current direction $\bx_k$
depends only on the previous one $\bx_{k-1}$. As a consequence, each random step
$\bx_{k-1} \rightarrow \bx_k$, for all $k\ge 1$,
are independent. This defines a discrete time random walk on the hypersphere ${\mathbb S}^{p-1}$.
It is of note that these steps are not necessarily identically
distributed.
However, an important case appears when they are isotropic, so that all the step
directions are equiprobable. Thus, the distribution of the $k$th-step direction
$\bx_k$ is rotationally symmetric about the previous one $\bx_{k-1}$, for all
$k\ge 1$. According to \eqref{eq:pdfsym},
its conditional {\em pdf} expresses as
\begin{align}
\label{eq:stepSym}
 f( \bx_k | \bx_{k-1} ) &= g_{k,k-1}(\bx_{k-1}^T \bx_k),
\end{align}
for all $k \ge 1$.
In the remainder, we assume that the initial direction $\bx_0$ is fixed to a deterministic direction $\bmu$
and that the walk is isotropic, which means that the distributions that govern each step express as \eqref{eq:stepSym}.
The following proposition states the link between the directional {\em pdf} after $n$ steps of the
isotropic random walk on
${\mathbb S}^{p-1}$ with the convolution on the double coset $L^1(SO(p-1)\backslash SO(p)/ SO(p-1), {\mathbb R})$ introduced in
section \ref{sec:coset}.

\begin{thm} %Mean direction and mean resultant length
\label{prop:multiconvpdfsp}
Given an isotropic $n$-step random walk on ${\mathbb S}^{p-1}$, the \emph{pdf} of
$\bx_n \in {\mathbb S}^{p-1}$ is the $n$-fold convolution
in $L^1(SO(p-1)\backslash SO(p)/ SO(p-1), {\mathbb R})$, where $SO(p-1)$ is the rotation subgroup such that
$\bmu$ is left invariant. It reads
\begin{align}
\label{eq:multiconv}
    f(\bx_n ; \bmu) = \left(g_{n,n-1} \star_{\bmu} \cdots \star_{\bmu} g_{1,0}\right) \  (\bx_n )
\end{align}
where $f(\bx_{k} | \bx_{k-1}) = g_{k,k-1}(\bx_{k}^T\bx_{k-1})$
can be identified as the conditional {\em pdf} of $\bx_k$ given $\bx_{k-1}$.
\end{thm}
\begin{proof}
The proof is conducted by induction. Since $\bx_0= \bmu$ is a  deterministic vector, the
{\em pdf} of $\bx_1$ is  $f(\bx_1 ; \bmu)=  g_{1,0}(\bmu^T\bx_1)$ and belongs to
$L^1(SO(p-1)\backslash SO(p)/ SO(p-1), {\mathbb R})$. Thus the base case holds for $n=1$.
Assume now that the pdf of $\bx_{k-1}$, for $k > 1$,  is symmetrical with
respect to $\bmu$, \emph{i.e.} $f(\bx_{k-1} ; \bmu )=g_{k-1}(\bmu^T \bx_{k-1})
\in L^1(SO(p-1)\backslash SO(p)/ SO(p-1), {\mathbb R})$,
and is given by the following $(k-1)$-fold convolution:
$
f_{k-1}(\bx_{k-1} ; \bmu) =  \left(g_{k-1,k-2} \star_{\bmu} \cdots \star_{\bmu} g_{1,0}\right) \  (\bx_{k-1}).
$

Due to the isotropic assumption, the conditional {\em pdf}  $f(\bx_{k} | \bx_{k-1})=  g_{k,k-1}(\bx_{k}^T\bx_{k-1})
$ also belongs to $L^1(SO(p-1)\backslash SO(p)/ SO(p-1), {\mathbb R})$. Moreover,
this conditional {\em pdf} allows us to express the density of $\bx_k$ as
\begin{align*}
f(\bx_{k} ; \bmu) &= \int_{{\mathbb S}^{p-1}} \hspace{-3mm} f(\bx_{k} | \bx_{k-1}) f(\bx_{k-1} ; \bmu)d \bx_{k-1},\\
&=  \int_{{\mathbb S}^{p-1}} \hspace{-3mm} g_{k,k-1}(\bx_{k}^T\bx_{k-1}) g_{k-1}(\bmu^T\bx_{k-1} ; )d \bx_{k-1}.
\end{align*}
According to \eqref{eq:conv}, we recognize the following convolution on the double coset:
$f(\bx_{k} ; \bmu) = ( g_{k,k-1} \star_{\bmu} g_{k-1} ) (\bx_{k})$. Thus  $f(\bx_{k} ; \bmu)$  is also symmetrical
about $\bmu$ according to property i) of Prop. \ref{prop:conv}. As $g_{k-1}$ is assumed to be a
$k-1$-fold convolution, it comes finally by associativity  that
$f(\bx_{k} ; \bmu)=\left(g_{k,k-1} \star_{\bmu} \cdots \star_{\bmu} g_{1,0}\right) \  (\bx_{k})$,
and the inductive step holds.
\end{proof}

A direct consequence of Theorem \ref{prop:multiconvpdfsp} is that
the distribution of $n$-step random walk is rotationally symmetrical about
the initial direction $\bmu$ since  $f(\bx_n ; \bmu) \in L^1(SO(p-1)\backslash SO(p)/ SO(p-1), {\mathbb R})$.
Furthermore, it allows us to express  the characteristic function of the random walk based on the Legendre polynomial
moments of each step.
\begin{corollary}[Mean and Fourier Coefficient of the isotropic random walk]
\label{cor:FourierCoeff}
For all $n\ge 1$, the mean of the $n$-step direction $\bx_n \in {\mathbb S}^{p-1}$ expresses as
\begin{align}
 E[ \bx_n ]= \left( \prod_{k=1}^n \rho_{k,k-1} \right)  \bmu,
 \label{eq:meann}
\end{align}
where  $\rho_{k,k-1}= E[ \bx_{k-1}^T \bx_k | \bx_{k-1} ]$ is the mean
resultant length for the conditional distribution $f(\bx_k | \bx_{k-1})=g_{k,k-1}(\bx_{k-1}^T \bx_k )$
that governs the $k$th step  $\bx_{k-1} \rightarrow \bx_k$.

More generally,  for all $n\ge 1$,  $\ell \ge 0$, the $\ell$th order Fourier coefficient of the
$n$-step distribution is
\begin{align}
 \widehat{f}^{\otimes n}_{\ell} \equiv E [ P_{\ell}( \bmu^T \bx_n ) ]=   \prod_{k=1}^n \widehat{g_{k,k-1}}_{\ell},
 \label{eq:FourierConvCoeff}
\end{align}
where $\widehat{g_{k,k-1}}_{\ell}= E[  P_{\ell}(\bx_{k-1}^T \bx_k) | \bx_{k-1} ]$ denotes the Fourier coefficient
for the conditional distribution $f(\bx_k | \bx_{k-1})=g_{k,k-1}( \bx_{k-1}^T \bx_k )$.
\end{corollary}
\begin{proof}
Eq. \eqref{eq:FourierConvCoeff} is derived from the multiconvolution formula \eqref{eq:multiconv} and the iterative use of
the convolution theorem \eqref{eq:FourierConv} given in Prop. \ref{prop:conv}. Eq. \eqref{eq:meann} 
is derived from the mean of a  rotationally symmetric distribution \eqref{eq:MomStep}  and from \eqref{eq:FourierConvCoeff} when $\ell=1$ since
$P_1(t)=t$ for all $p \ge 2$.
\end{proof}
Eq. \eqref{eq:meann} shows that the mean direction of the $n$-step is the
initial direction $\bmu$, while its mean
resultant length reduces to the product of  the mean resultant  lengths
associated with each step:
\begin{align}
\rho_n \equiv E[\bmu^T \bx_n ] = \prod_{k=1}^n \rho_{k,k-1}.
\label{eq:rhon}
\end{align}
This formula emphasizes that the directional dispersion increases with the number $n$ of steps,
since $0\le \rho_{k,k-1} \le 1$ for all $k \ge 1$.

Based on the Fourier coefficient formulas, it becomes possible to obtain a Fourier expansion
of the rotationally symmetric random walk {\em pdf}.
\begin{corollary}[Fourier expansion of the isotropic random walk {\em pdf}]
 For all $n\ge 1$, the Fourier expansion  {\em pdf} of the $n$-step direction
 $\bx_n \in {\mathbb S}^{p-1}$ reads
\begin{align}
\label{eq:fourierexpRW}
f(\bx_n ;\bmu) &= \sum_{\ell \geq 0} c_{p,\ell} \widehat{f}^{\otimes n}_{\ell} P_{\ell}( \bmu^T \bx),
\end{align}
where the Fourier coefficients $\widehat{f}^{\otimes n}_{\ell}$ are given in \eqref{eq:FourierConvCoeff},
the normalizing constants $c_{p,\ell}$  being defined in \eqref{eq:normalcst}.
\end{corollary}
\begin{proof}
The distribution of $\bx_n$ is symmetric according to Theorem \ref{prop:multiconvpdfsp},
and the Fourier expansion formula \eqref{eq:fourierexp} can be applied.
\end{proof}

It is important to note that when all the steps of the random walk are identically distributed, i.e.
$f( \bx_k |  \bx_{k-1}) = g(\bx_{k-1}^T \bx_k)$ for all $k\ge 1$, the Fourier coefficients reduces to
\begin{align}
 \widehat{f}^{\otimes n}_{\ell} =   \left( \widehat{g}_{\ell}\right)^n,
 \label{eq:FourierConvCoeffIdent}
\end{align}
where $\widehat{g}_{\ell}= E[  P_{\ell}(\bmu^T \bx_1) ]$ is the Fourier coefficient of
the distribution that governs a random walk step.

Finally, it is possible to derive sufficient conditions to ensure the unimodality of the isotropic random walk
\begin{corollary}[Unimodality of the isotropic random walk]
\label{cor:uniRW}
Assume that the conditional directional distributions
$f(\bx_k |\bx_{k-1})=g_{k,k-1}(\bx_{k-1}^T \bx_k)$,
which are rotationally symmetric, are also absolutely continuous (or equivalently that $g_{k,k-1}$
is a continuous function on $[-1,1]$) and unimodal with a mode that
necessarily equals their mean direction.
Then, for all $n\ge 1$, the {\em pdf} of the $n$-step direction
 $\bx_n \in {\mathbb S}^{p-1}$ is also unimodal with  a mode that
 equals the original direction $\bmu \equiv \bx_0$.
 \end{corollary}
\begin{proof} This result is derived from the
multiconvolution formula \eqref{eq:multiconv} and the
iterative use of Theorem \ref{thm:uni}.
\end{proof}

%===========================%
%							%
% MULTIPLE SCATTERING MODEL %
%							%
%===========================%

\section{Multiple scattering model}
\label{sec:mulscatmod}
The model presented in this section is motivated by the description of multiple scattering which occurs in a wide range of applications in Physics and Engineering, including optical, microwave, acoustics or elastic waves \cite{Ishimaru1999} . We consider the description of the
distribution of the output direction of propagation of a particle/wave that propagated through a random medium.
This random medium is made of an homogeneous medium/matrix containing some inclusions of size of the same order
as the particle size (or wavelength of the wave). Inclusions have different physical properties inducing that a
scattering event happens each time the particle/wave encounters an inclusion (also named scatterer).
The number and locations of the scatterers is random and between two scattering events,
the wave/particle propagates balisticaly. It is a classical approach to consider that the time
between two scattering events follows an exponential law. Such an assumption leads to Compound
Poisson Process models as described in \cite{Ning,LeBihanMargerin,Said,Chatelain2013}. Here, we consider the more
general case where the intensity of the counting process is a random
process itself. We will thus make use of Compound Cox Processes taking values on hyperspheres to model multiple scattering.

\subsection{Compound Cox process on ${\mathbb S}^{p-1}$}

Consider an initial vector $\bx_0 \equiv \bmu \in {\mathbb S}^{p-1}$. After a time
$t$ (the time spent propagating in the random medium), assume the resulting vector $\bx_t \in {\mathbb S}^{p-1}$ is a
mixture made of contributions
of rotated versions of $\bmu$ an arbitrary number of times $n$. The weight of each contribution
is simply the probability that the wave/particle encountered $n$ scatterers during the period of time $t$,
{\em i.e.} ${\mathbb P}[N(t)=n]$. $N(t)$ is called the \emph{counting process}.
In the classical compound Poisson process \cite{Said}, $N(t)$ is an homogeneous Poisson process
and each individual weight is equal to $e^{-\lambda t}(\lambda t)^n/n!$ where the constant $\lambda$ is the
Poisson intensity parameter. This weight is
obtained when the time between two rotations of the vector is chosen to
have an exponential distribution with parameter $\lambda$. The equivalent
Poisson parameter $\lambda_t=\lambda t$ of $N(t)$ consists of  the mean number of rotation
events in the elapsed time $t$. In Physics, it is related to the {\em mean free
path $\ell$} like $\ell=c/\lambda$ where $c$ is the celerity in the medium. Thus,
$\ell$ is the mean distance between two consecutive rotation events (see \cite{LeBihanMargerin}).

Now, if the counting process $N(t)$ is no more a homogeneous Poisson process,
an alternative is to consider that the intensity measure of $N(t)$ is a random process $\Lambda(t)$.
$N(t)$ is then called a mixed Poisson process, or a Cox process \cite{Lefebvre2006}.

In this case, the distribution of $N(t)$ is a mixed Poisson distribution which reads \cite{Saleh1978}:
\begin{align}
{\mathbb P}[N(t)=n] & = {\cal P}_n\left[f_{\Lambda(t)}\right] \\
 & = \int_{0}^{+\infty} \frac{e^{-\lambda_t} \lambda_t^n}{n!} f_{\Lambda(t)}(\lambda_t) d\lambda_t
\end{align}
where ${\cal P}_n\left[f_{\Lambda(t)}\right]$ is called the \emph{Poisson transform} of the
mixing {\em pdf} $f_{\Lambda(t)}$ \cite{Saleh1978}.

In the isotropic case, the steps associated with the scattering events are governed by rotationally
symmetric distribution as explained in section \ref{sec:RW}.
For the sake of simplicity, all these random steps are assumed to be identically
distributed, and we denote as $\widehat{f}_\ell \equiv \widehat{g}_\ell$ the $\ell$th order Fourier coefficient
of the random step distribution, for all $\ell\ge 0$. This yields that
the $\ell$th order Fourier coefficient of the $n$-step random walk direction reduces to
$\widehat{f}^{\otimes n}_{\ell} =   \left( \widehat{f}_{\ell}\right)^n$
according to \eqref{eq:FourierConvCoeffIdent}.

Conditioning by the number of scattering events,
one gets the expression of the density of $\bx_t$:
\begin{align}
\begin{split}
f(\bx_t;\bmu)
=& {\cal P}_0\left[f_{\Lambda_t}\right] \delta_{\bmu}(\bx_t)  \\
&  + \sum_{n \geq 1} {\cal P}_n\left[f_{\Lambda_t}\right] f^{\otimes n}(\bx_t; \bmu ),
\label{eq:ftimet}
\end{split}
\end{align}
\noindent where
$\delta_{\bmu}(\bx_t)$ denotes a mass located in the original direction $\bmu \equiv \bx_0 \in {\mathbb S}^{p-1}$, and $f^{\otimes n}(\bx_t;\bmu )$ denotes the $n$-step random walk {\em pdf} with
original direction $\bmu$.

The {\em pdf} $f(\bx_t;\bmu)$ thus consists of a
mixture of $n$-fold convolutions of identical distributions according to Thm. \ref{prop:multiconvpdfsp},
for all $n>0$, plus a mass in $\bmu$ which corresponds to direct paths.

Equation (\ref{eq:ftimet}) covers several cases. If $N(t)$ is a homogeneous Poisson process
with intensity parameter $\lambda$, then
${\cal P}_n\left[f_{\Lambda_t}(w_t)\right]=\frac{e^{-\lambda t}(\lambda t)^n}{n!}$
and this case was considered in \cite{Chatelain2013,LeBihanMargerin,Said}.
If $N(t)$ is a Cox process, an interesting case appears when
its distribution belongs  to an exponential family.
As explained in \cite{Ferrari2007}, this happens when $\Lambda(t)$
is a stationary i.i.d. Gamma process , {\em i.e.}
$\Lambda_t \sim \mathcal{G}(t;\xi,\theta) \sim \mathcal{G}(\xi t,\theta)$ where $\xi$
is the shape parameter and $\theta$ the scale parameter.
The distribution of this process is thus:
\begin{equation}
f_{\Lambda_t}(\lambda_t=x) = \frac{\theta^{\xi t}}{\Gamma(\xi t)} x^{\xi t -1} e^{-\theta x},
\end{equation}
for all $x>0$.
It comes by direct calculation that in this case the Poisson transform is
\begin{equation}
{\cal P}_n\left[f_{\Lambda_t}\right]=\frac{\Gamma(n+\xi t)}{ n! \Gamma(\xi t)}
\frac{\theta^{\xi t}}{(\theta+1)^{n+\xi t}}
\label{eq:PoissonT_gammapdf}
\end{equation}
which shows that when $\Lambda_t$ is a Gamma process with scale parameter $\xi$ and shape parameter $\theta$,
$N(t)$ is a negative binomial process.
In fact, in such case, the weight coefficients in (\ref{eq:ftimet}) follow a negative binomial law
$\mathcal{NB}(r_t,q)$ with stopping-time parameter $r_t=\xi t$ and success
probability $q= (\theta+1)^{-1}$.

\subsection{Characteristic function of the multiple scattering process}
\label{subsec:charfmsprocess}

The distribution of the multiple scattering process  \eqref{eq:ftimet} consists of the mixture of a mass in
$\bmu$ for the direct paths and a continuous distribution that consists of the $n$-fold random step
convolutions for all $n\ge 1$.
This continuous distribution is the conditional distribution of $\bx_t$ given there is at least one diffusion, i.e.
$N(t)>0$.
Its  {\em pdf} denoted as $f^{\otimes>0}$  reads $f^{\otimes>0}(\bx_t; \bmu ) = c_0 h^{\otimes>0}(\bx_t; \bmu )$
where $c_0= \left(1-{\cal P}_0\left[f_{\Lambda_t}\right] \right)^{-1}$ is the normalizing constant of the truncated distribution corresponding to the event $N(t)>0$ and $h^{\otimes>0}$ is the following unnormalized density
\begin{align}
 h^{\otimes>0}(\bx_t; \bmu ) &=
 \sum_{n \geq 1} {\cal P}_n\left[f_{\Lambda_t}\right] f^{\otimes n}(\bx_t; \bmu ).
\label{eq:pdfPosDiff}
 \end{align}
This {\em pdf} is rotationally symmetric about $\bmu$ as a mixture of symmetric distributions.
This shows that $h^{\otimes>0}(\bx_t; \bmu)$ admits a Fourier expansion \eqref{eq:fourierexp}.
Note also that when the distribution that governs the random steps is unimodal, the {\em pdf} of $\bx_t$ is unimodal with mode $\bmu$ as a mixture of unimodal distributions according to Corollary \ref{cor:uniRW}.

Considering the Laplace transform of the mixing process $\Lambda_t$
\begin{align*}
 \mathcal{L}_{\Lambda_t}[z]=E[e^{-z \Lambda_t}]= \int_{0}^{+\infty} e^{-z \lambda_t }
 f_{\Lambda(t)}(\lambda_t) d\lambda_t,
\end{align*}
we obtain the following expression of the Fourier coefficients.
\begin{lemma}
\label{lemma:DiffFourierPos}
The Fourier coefficients of the continuous unnormalized density $h^{\otimes>0}$ express as
 \begin{align}
  \widehat{h^{\otimes>0}}_{\ell} &=
  \mathcal{L}_{\Lambda_t}\left[1-\widehat{f}_{\ell}\right] - \mathcal{L}_{\Lambda_t}\left[1\right],
  \label{eq:FourierDiff}
 \end{align}
 for all $\ell \ge 0$, where $\widehat{f}_{\ell}$ is the Fourier coefficient of
 the isotropic and identically distributed random steps % given in \eqref{eq:legVMF} in the vMF case,
 and where $\mathcal{L}_{\Lambda_t}\left[1\right]=  {\cal P}_0 \left[ f_{\Lambda_t} \right]= \Pr(N_t=0)$.
\end{lemma}
\begin{proof} Based on the orthogonality property of the Legendre polynomials,
it comes from \eqref{eq:pdfPosDiff} and the Fourier expansion  of each $n$-fold
convolved {\em pdf} $f^{\otimes n}$ \eqref{eq:fourierexpRW}, for $n\ge 1$, that
  \begin{align}
  \widehat{h^{\otimes>0}}_{\ell} &=  \sum_{n \ge 1} {\cal P}_n\left[f_{\Lambda_t}  \right]  \widehat{f}_{\ell}^n,
  \label{eq:FourierDiffPos}
  \end{align}
  where $|\widehat{f}_{\ell}|\le 1$ by construction. According to Fubini theorem, one can interchange
  the summation symbol with the Poisson transform integral. This yields
  \begin{align*}
  \widehat{h^{\otimes>0}}_{\ell} &= \int_{0}^{+\infty} e^{-\lambda_t} f_{\Lambda_t}(\lambda_t) \sum_{n \ge 1}
  \frac{ \left( \lambda_t \widehat{f}_{\ell}\right)^n}{n!} d\lambda_t,\\
  &=  \int_{0}^{+\infty} e^{-\lambda_t} f_{\Lambda_t} (\lambda_t)
  \left[ e^{ \lambda_t \widehat{f}_{\ell}} -1 \right] d\lambda_t,\\
  &=  \mathcal{L}_{\Lambda_t}\left[1-\widehat{f}_{\ell}\right] - \mathcal{L}_{\Lambda_t}\left[1\right],
  \end{align*}
  for all $\ell\ge 0$. Note finally that the Laplace transform $\mathcal{L}_{\Lambda_t}[z]$ is well-defined for all
  $z\ge 0$ since $\Lambda_t$ is a positive random variable. Thus the Fubini theorem holds as $1-|\widehat{f}_{\ell}|\ge 0$,
  and the Fourier coefficients are well defined.
\end{proof}
Consider now the probability generating function of the Cox process $N_t$:
\begin{align*}
  G_{N_t}[z]=E\left[ z^{N_t} \right]= \sum_{n\ge 0 } {\cal P}_n \left[ f_{\Lambda_t} \right] z^n.
\end{align*}
A classical result about mixed Poisson distribution, see for instance \cite{Ferrari2007,Chatelain2009}, is that
$G_{N_t}$ can be easily
derived from the Laplace transform of the mixing distribution as
\begin{align*}
 G_{N_t}[z]= \mathcal{L}_{\Lambda_t}\left[1-z\right].
\end{align*}
The Fourier coefficients given in Lemma \ref{lemma:DiffFourierPos} can thus be expressed in an equivalent way as
 \begin{align}
  \widehat{h^{\otimes>0}}_{\ell} &=
G_{N_t}\left[\widehat{f}_{\ell}\right] - G_{N_t}\left[0\right],
  \label{eq:FourierDiffPosGen}
  \end{align}
where $ G_{N_t}\left[0\right]= \mathcal{L}_{\Lambda_t}\left[1\right]=  {\cal P}_0 \left[ f_{\Lambda_t} \right]= \Pr(N_t=0)$.

Finally, based on the Fourier expansion of the continuous function  $h^{\otimes>0}$, we obtain the following results for the distribution of $\bx_t$.
\begin{prop}
\label{prop:FourierScat}
The {\em pdf} of the direction $\bx_t$ in the multiple scattering process can be expanded as
\begin{align}
f(\bx_t ; \bmu)
=& {\cal P}_0\left[f_{\Lambda_t}\right] \delta_{\bmu}(\bx_t)  +
\sum_{\ell \ge 0 } c_{p,\ell} \widehat{h^{\otimes>0}}_{\ell} P_{\ell} (\bmu^T \bx_t),
 \label{eq:pdfexpScat}
\end{align}
where the coefficients $\widehat{h^{\otimes>0}}_{\ell}$ are given in
\eqref{eq:FourierDiff}, or equivalently in \eqref{eq:FourierDiffPosGen}.\\
Moreover, the Legendre polynomial moments read
%, i.e. the Fourier coefficients of $\bx_t$ read
\begin{align}
E[P_\ell( \bmu^T \bx_t )] &= G_{N_t}\left[\widehat{f}_{\ell}\right] =
 \mathcal{L}_{\Lambda_t}\left[1-\widehat{f}_{\ell}\right],
 \label{eq:scatLeg}
\end{align}
for all $\ell \ge 0$, where $\widehat{f}_{\ell}$ is the Fourier coefficient of
the isotropic and identically distributed random steps.
\end{prop}
\begin{proof}
 It remains to show the Legendre polynomial moment formula \eqref{eq:scatLeg}.
 The distribution of $\bx_t$ is a linear mixture of a mass in $\bmu$ with probability
 ${\cal P}_0\left[f_{\Lambda_t}\right]=  G_{N_t}\left[0\right]$
 and a continuous distribution.
 As $P_\ell( \bmu^T \bmu )= P_\ell( 1 )=1$, the contribution of the mass in $\bmu$ to the $\ell$th order Legendre
 moment reduces to its probability.
 By linearity, it comes that
 $E\left[P_\ell( \bmu^T \bx_t )\right]= G_{N_t}\left[0\right] +  \widehat{h^{\otimes>0}}_{\ell} = G_{N_t}\left[\widehat{f}_{\ell}\right]$
 according to \eqref{eq:FourierDiffPosGen}.
\end{proof}
It is important to note that the distribution of $\bx_t$ is not continuous due to the mass in $\bmu$.
As written in \eqref{eq:pdfexpScat}, only the continuous part admit a Fourier expansion in the pointwise convergence sense.
However, all the Legendre polynomial moments of the multiple scattering direction $\bx_t$
are well-defined and have  tractable expressions given in \eqref{eq:scatLeg}.

One consequence is that the mean resultant length of the multiple scattering process $\bx_t$ 
express as $\rho_t= E[ \bmu^T \bx_t ]= G_{N_t}\left[\widehat{f}_{1}\right]$. And due to the rotational 
symmetry, the mean of $\bx_t$ is
\begin{align*}
 E[\bx_t] &= \rho_t \bmu =  G_{N_t}\left[\widehat{f}_{1}\right] \bmu.
\end{align*}

Proposition \ref{prop:FourierScat} allows us to express the Fourier
coefficients and the Legendre polynomial moments in two particular cases
that we are especially interested in. First in the case of the homogeneous Poisson process where $\Lambda_t= \lambda_t$ is
a deterministic constant, it comes that
 \begin{align}
 \begin{split}
  \widehat{h^{\otimes>0}}_{\ell} 
  &= e^{-\lambda_t}  \left[ e^{ \lambda_t  \widehat{f}_{\ell} } -1 \right],\\
  E\left[P_\ell( \bmu^T \bx_t )\right] &=   e^{ -\lambda_t \left[ 1 - \widehat{f}_{\ell} \right] },
  \label{eq:FourierDiffPoiss}
  \end{split}
 \end{align}
 for all $\ell \ge 0$.
Second, when $\Lambda_t~\sim \mathcal{G}(\xi t, \theta)$ is a Gamma process, i.e. that
$N_t~\sim \mathcal{NB}(r_t,q)$ is a Negative Binomial process with $r_t=\xi t$ and $q=(\theta+1)^{-1}$,
straightforward
computations lead to
 \begin{align} \begin{split}
  \widehat{h^{\otimes>0}}_{\ell} &= (1-q)^{r_t}
  \left[ \left( 1 - q  \widehat{f}_{\ell} \right)^{-r_t} -1  \right],\\
  E\left[P_\ell( \bmu^T \bx_t )\right] &=   \left( \frac{ 1-q }{ 1 - q  \widehat{f}_{\ell} } \right)^{r_t},
  \label{eq:FourierDiffGam}
  \end{split}
 \end{align}
 for all $\ell \ge 0$.

\section{von Mises-Fisher random walk and multiple scattering process on ${\mathbb S}^{p-1}$}
\label{sec:VMFRWSP}
\subsection{von Mises-Fisher distribution on ${\mathbb S}^{p-1}$}
\label{sec:vMF}

The von Mises-Fisher distribution \cite[p. 167]{Mardia} is probably
the most important distribution in the
statistics of hyperspherical data and plays a role on ${\mathbb S}^{p-1}$
analogue to the role of the normal distribution on the real line.
This distribution, denoted as $M_p(\bmu,\kappa)$, is defined by the following
{\em pdf} for all $\bx \in {\mathbb S}^{p-1}$
\begin{align}
\label{eq:VMF}
f(\bx ; \bmu , \kappa)=
\frac {\kappa^{p/2-1}} {(2\pi)^{p/2}I_{p/2-1}(\kappa)}
e^{\kappa \bx^T \bmu  },
\end{align}
where $I_{\nu}(\cdot)$ is the modified Bessel function \cite[p.
374]{Abramowitz}, $\bmu \in {\mathbb S}^{p-1}$ corresponds to the mean direction
and $\kappa \ge 0 $ is the {\it concentration parameter}: the larger the value
of $\kappa$, the more concentrated is the distribution about the mean direction
$\bmu$. Conversely, when
$\kappa=0$ the distribution reduces to the uniform distribution on ${\cal
S}^{p-1}$. Its mean resultant length takes the form:
\begin{align}
\label{eq:rhoVMF}
 \rho &\equiv A_p(\kappa) = \frac{I_{p/2}(\kappa)}{I_{p/2-1}(\kappa)},
\end{align}
which reduces to
\begin{align*}
 \rho &\equiv A_3(\kappa) = \coth{\kappa} - \frac{1}{\kappa},
\end{align*}
when $p=3$.
Based on the {\em pdf} expression \eqref{eq:VMF}, it is straightforward to
see that this distribution is rotationally symmetric and unimodal with mode
$\bmu$ when $\kappa>0$.
Finally, as explained in \cite{Kent78}, the characteristic function of the von Mises Fisher distribution
$f(\bx ; \bmu , \kappa)$ takes the form:
\begin{equation}
\label{eq:legVMF}
 \widehat{f}_{\ell}(\kappa) = E\left[ P_{\ell}(\bx^T\bmu)\right] = \frac{I_{\ell+\nu}(\kappa)}{I_{\nu}(\kappa)},
\end{equation}
where $\nu=p/2-1$, for $\kappa>0$ and $\ell\ge 0$.

\subsection{von Mises-Fisher random walk}

One problem of characterizing more deeply and inferring efficiently
the distribution of the $n$th-step direction $\bx_n$
is that, except the Fourier series expansion \eqref{eq:fourierexpRW},
there is no simple closed form expression
of the density of the multiply convolved distribution \eqref{eq:multiconv}.

However, when all the isotropic random walk steps are governed by
unimodal distributions, Corollary \ref{cor:uniRW}
says that the $n$-step random walk direction  $\bx_n$ is also governed
by an unimodal rotationally symmetric distributions with mode the original
direction $\bmu \equiv \bx_0$. This suggests that the distribution of $\bx_n$
can be well fitted by a standard unimodal rotationally symmetric distributions
with mode $\bmu$. Due to the properties of the vMF distributions presented in
section \ref{sec:vMF}, this family seems to be a good candidate to
fit the distribution of $\bx_n$, for $n\ge 1$. It leads to model the
distribution of $\bx_n$ by a $M_p(\bmu,\tilde{\kappa}_n)$  distribution,
where $\tilde{\kappa}_n$ is an equivalent concentration parameter
for the $n$-step direction.

Moreover, we consider now that all the random walk steps \eqref{eq:stepSym}
are identically distributed according to a vMF distribution with concentration
parameter $\kappa$, that is
\begin{align*}
 \bx_k | \bx_{k-1} \sim M_p(\bx_{k-1},\kappa),
\end{align*}
and the resulting random walk is called the {\em vMF random walk} with concentration parameter $\kappa$.
This random walk is isotropic and unimodal with mode $\bmu \equiv \bx_0$
according to Corollary \ref{cor:uniRW}. It is possible to obtain, in the high concentration case,
a simple vMF asymptotic distribution for the
$n$-step direction.
\begin{thm}
\label{thm:vMFapprox}
Consider the  vMF random walk with concentration parameter $\kappa$.
Then, in the large $\kappa$ and  small $n\ge 1$ case, i.e. when $n / \kappa \rightarrow 0$,  $\bx_n$ is asymptotically distributed as
$M_p(\bmu,\tilde{\kappa}_n)$ where
\begin{align}
  \label{eq:kappaEqThm}
  \tilde{\kappa}_n &= \frac{ \kappa -1/2 }{ n } + 1/2
\end{align}
is the equivalent concentration parameter. The asymptotic distribution 
yields a third-order approximation of the Fourier coefficients
\begin{align*}
 \widehat{f}^{\otimes n}_{\ell}= \widetilde{f}^n_{\ell} +  O \left( \left(\frac{n}{\kappa}\right)^3 \right), \quad 
 \textrm{ as } \ \frac{n}{\kappa} \rightarrow 0,
\end{align*}
for any $\ell \ge 0$, where
$\widehat{f}^{\otimes n}_{\ell}$ is the $\ell$th-order Fourier coefficient of the $\bx_n$ distribution,
while $\widetilde{f}^n_{\ell}$ denotes the  $\ell$th-order Fourier coefficient of the asymptotic distribution
$M_p(\bmu,\tilde{\kappa}_n)$.
\end{thm}
\begin{proof} Both $M_p(\bmu,\tilde{\kappa}_n)$ and the distribution of $\bx_n$ are rotationally symmetric about
$\bmu$. Thus it is sufficient to show the asymptotic equivalence of their tangent part about $\bmu \in {\mathbb S}^{p-1}$.
The tangent part is bounded in $[-1,1]$, which ensures that its moments are well-defined and also belong to  $[-1,1]$.
Thus its distribution is uniquely defined by its moments, see for instance \cite[Theorem 30.1, p. 388]{Billingsley95}.
As the family of Legendre polynomials form a polynomial basis, we can conclude by the method of moments that
the distributions are asymptotically equivalent if their Fourier coefficients are asymptotically equivalent.

%Let $\widehat{f}^{\otimes n}_{\ell}$ be the $\ell$th-order Fourier coefficient of the $\bx_n$ distribution,
%while $\widetilde{f}^n_{\ell}$ denotes the  $\ell$th-order Fourier coefficient of $M_p(\bmu,\tilde{\kappa}_n)$.
For $n\ge 1$, we obtain from \eqref{eq:FourierConvCoeffIdent} and \eqref{eq:legVMF}, and from \eqref{eq:legVMF} and \eqref{eq:kappaEqThm}
respectively,  that
\begin{align*}
 \widehat{f}^{\otimes n}_{\ell}=  \left( \frac{I_{\ell+\nu}(\kappa)}{I_{\nu}(\kappa)} \right)^n, \qquad & 
 \widetilde{f}^n_{\ell}  =  \frac{I_{\ell+\nu}(\tilde{\kappa}_n)}{I_{\nu}(\tilde{\kappa}_n)}.
\end{align*}
Using now the following asymptotic expansion of the modified Bessel function for large
$\kappa$  \cite[p. 377]{Abramowitz}:

{\small
\begin{align*}
 I_{\nu}(\kappa)=& \frac{e^{\kappa}}{ \sqrt{2\pi\kappa} } \left[
 1 - \frac{4 \nu^2 -1}{8 \kappa} +  \frac{(4 \nu^2 -1)(4 \nu^2 - 9)}{2! (8 \kappa)^2} +
 O\left(\frac{1}{\kappa^3} \right) \right],
\end{align*}}
yields that, for any  $\ell\ge 0$, 
{\small
\begin{align*}
 \widehat{f}^{\otimes n}_{\ell}
 =& 1 - \frac{ \ell  n (\ell + 2 \nu) }{2\kappa} +
\frac{  \ell n  (\ell + 2 \nu) (n \ell^2 + 2 n \nu \ell - 2) } {8 \kappa^2}  + O\left(\frac{n^3}{\kappa^3}\right),
\end{align*}
}
and that 

{\small
\begin{align*}
\widetilde{f}^n_{\ell} 
 =& 1 - \frac{ \ell  n (\ell + 2 \nu) }{2\kappa} +
\frac{  \ell n  (\ell + 2 \nu) (n \ell^2 + 2 n \nu \ell - 2) } {8 \kappa^2} + O\left(\frac{n^3}{\kappa^3}\right).
\end{align*}
}
This shows that
$
\widehat{f}^{\otimes n}_{\ell}= \widetilde{f}^n_{\ell} +  O \left( \right(\frac{n}{\kappa}\left)^3 \right)$ 
as $\frac{n}{\kappa}$ tends to zero, for any $\ell\ge 0$,  which concludes the proof.
\end{proof}
Note that the expression \eqref{eq:kappaEqThm} of the equivalent concentration parameter
has been derived in \cite{Chatelain2013} by matching the mean resultant length in the asymptotic case.
Theorem \ref{thm:vMFapprox} shows that a similar result extends to all the Legendre polynomial moments, and thus to
the distribution.

To appreciate the accuracy of the vMF approximation given by Theorem \ref{thm:vMFapprox},
Fig. \ref{fig:asymptconv}
compares the distribution {\em pdf}s and quantiles of the random walk tangent part $t=\bmu^T\bx_n$
with the asymptotic one for $n=10$ steps on the $p=3$ dimensional sphere. As explained in
\cite[p. 168--170]{Mardia}, the {\em pdf} of the tangent part $t$ can be derived from the symmetric
directional {\em pdf}
$f(\bx)= g(\bmu^T\bx)$ on ${\mathbb S}^{p-1}$ as
\begin{align*}
f_p(t) &=
\omega_{p-1} B\left( \frac{p-1}{2},\frac{1}{2} \right)^{-1} g(t) \; (1-t^2)^{\frac{p-3}{2}},
\end{align*}
for all $-1 \le t \le 1$, where $B(\cdot,\cdot)$ is the classical Beta function.
In dimensions $p=3$, the Fourier expansion \eqref{eq:fourierexpRW} leads to
the following simple expression for the projected distribution density,
i.e. the tangent part {\em pdf}
\begin{align*}
f_3(t) &= \sum_{\ell \ge 0 }\frac{2\ell + 1 }{2}\widehat{g}_{\ell}  P_\ell(t),
\end{align*}
for all $-1 \le t \le 1$,
with $\widehat{g}_{\ell}$ the Fourier coefficient of the directional distribution {\em pdf} $g$.
According to \eqref{eq:FourierConvCoeffIdent}
and \eqref{eq:legVMF}, $\widehat{g}_{\ell}=\widehat{f}^{\otimes n}_{\ell}=
\left( \frac{I_{1/2+\ell}(\kappa)}{I_{1/2}(\kappa)}\right)^n$ for the $n$-step random walk on  ${\mathbb S}^{2}$.
This Fourier expansion allows us to numerically evaluate the exact projected {\em pdf}.
The empirical quantiles are estimated from  $10^7$  Monte-Carlo runs.

As expected, Fig. \ref{fig:LOW} shows that for low concentration, the asymptotic
distribution diverges from the real one.
However for high enough concentration, Figs \ref{fig:MEDIUM} and especially \ref{fig:HIGH} 
show the accuracy of the asymptotic approximation. In practice, when $p=3$ the asymptotic distribution
gives a reasonable approximation when $\tilde{\kappa}_n \approx \frac{\kappa}{n} \ge 10$ and is accurate
when $\tilde{\kappa}_n \approx \frac{\kappa}{n} \ge 50$.
\begin{figure}[htbp!]
\centering
\subfigure[$\kappa=10\ (\rho \approx 0.9$)\label{fig:LOW}]{\centering
\includegraphics[width=.5\linewidth]{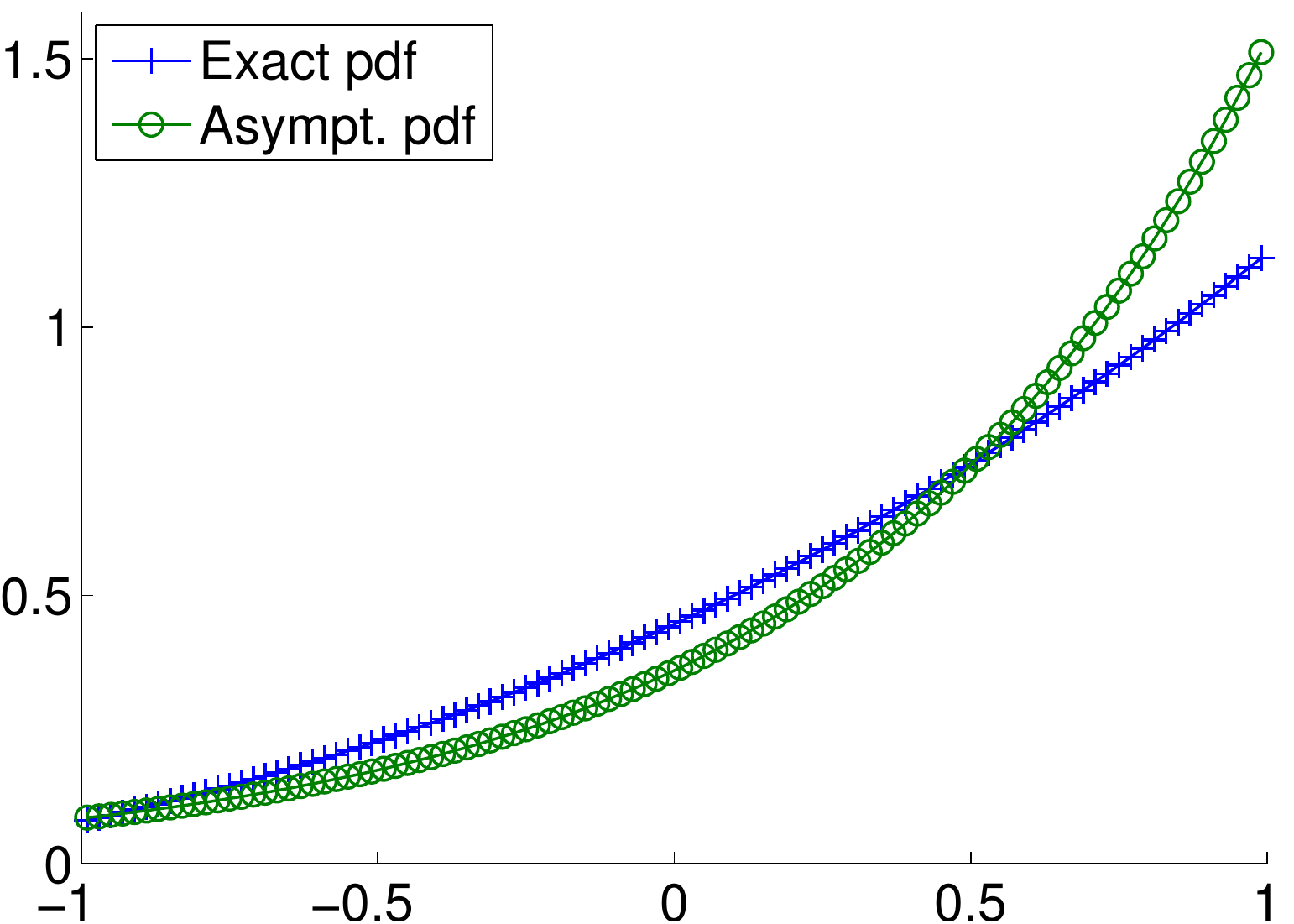}
\hfill
\includegraphics[width=.5\linewidth]{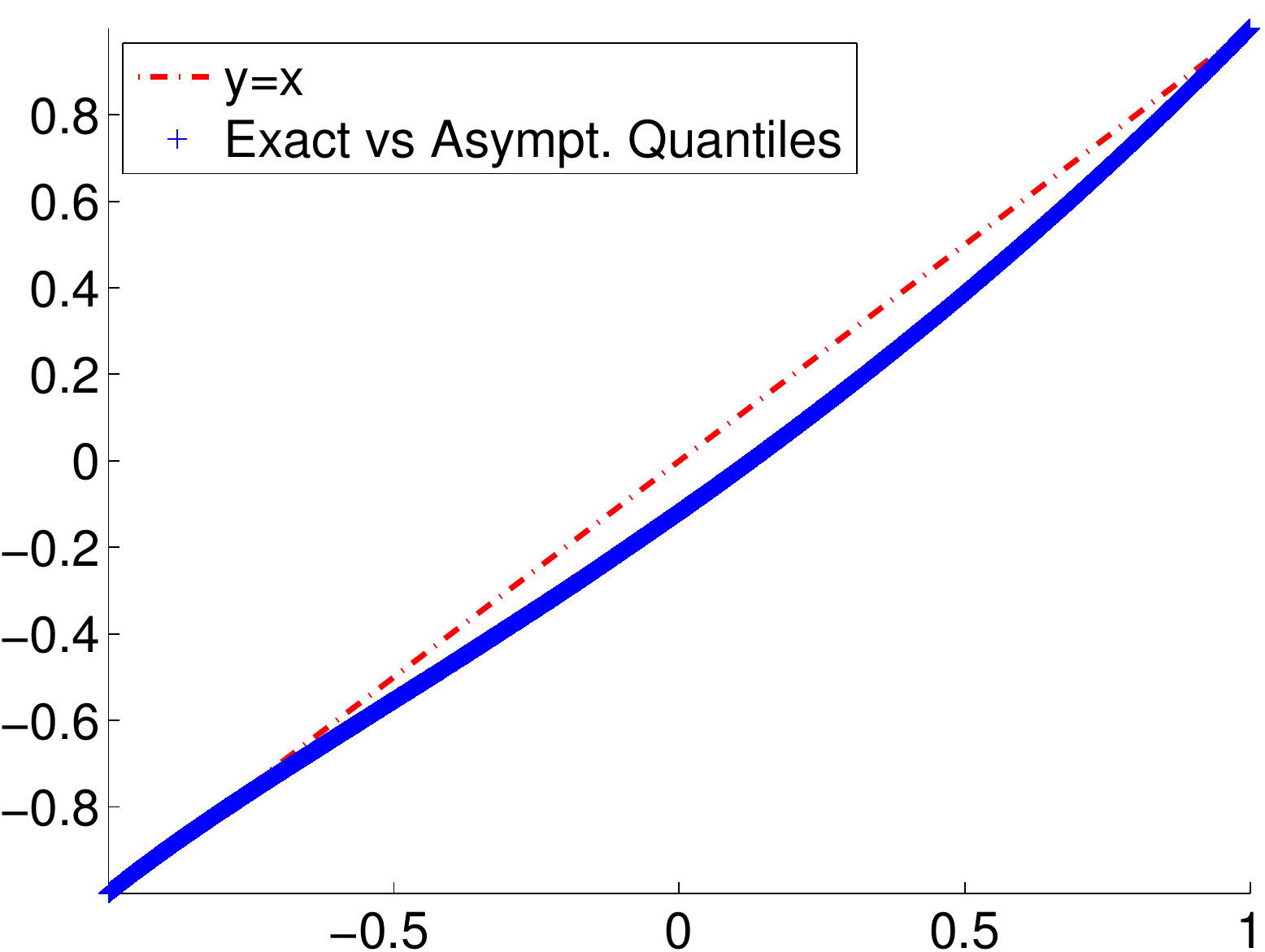}}\\
\subfigure[$\kappa=100\ (\rho \approx 0.99$)\label{fig:MEDIUM}]{\centering
\includegraphics[width=.5\linewidth]{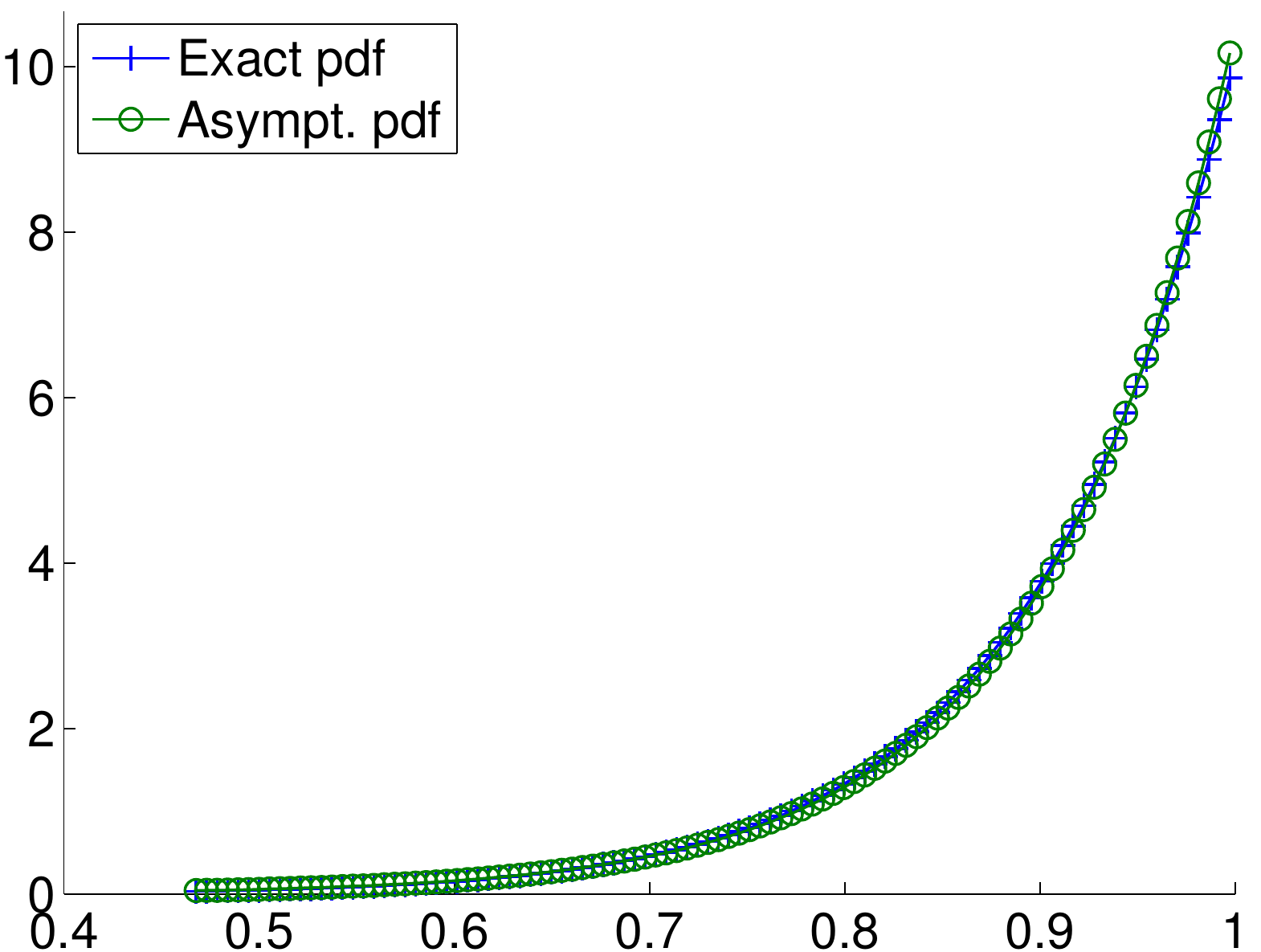}
\hfill
\includegraphics[width=.5\linewidth]{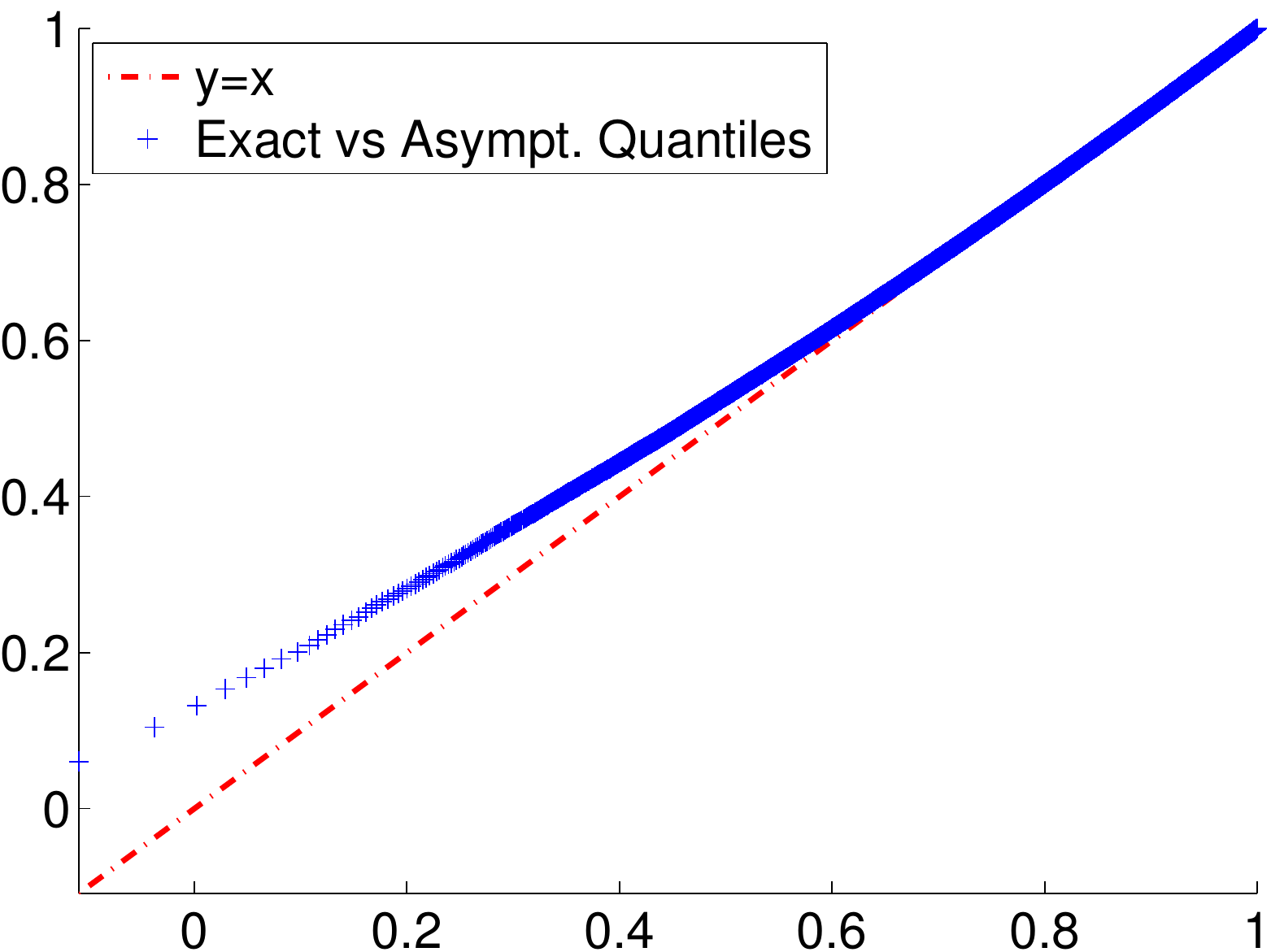}} \\
\subfigure[$\kappa=1000\ (\rho \approx 0.999$)\label{fig:HIGH}]{\centering
\includegraphics[width=.5\linewidth]{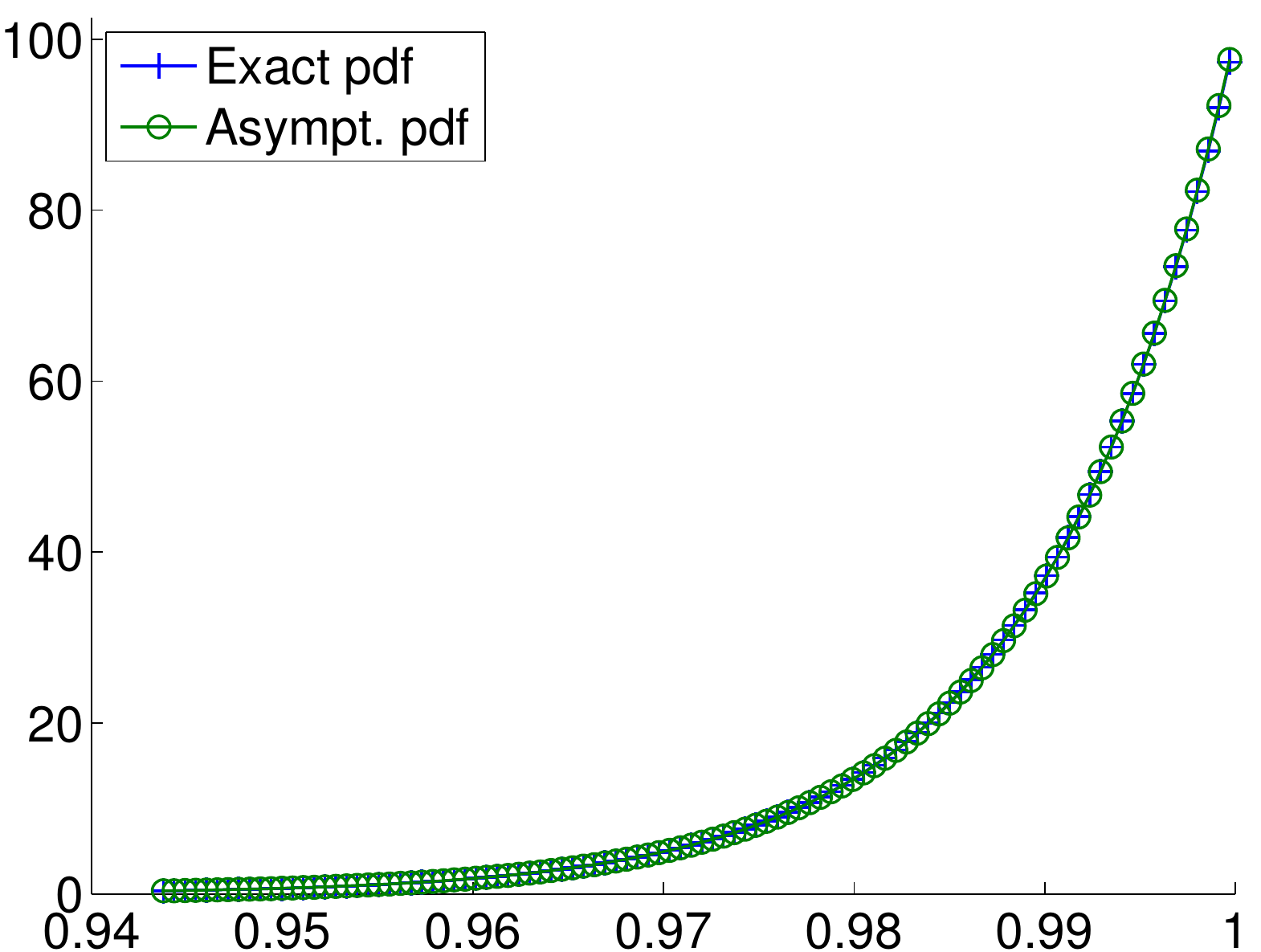}
\hfill
\includegraphics[width=.5\linewidth]{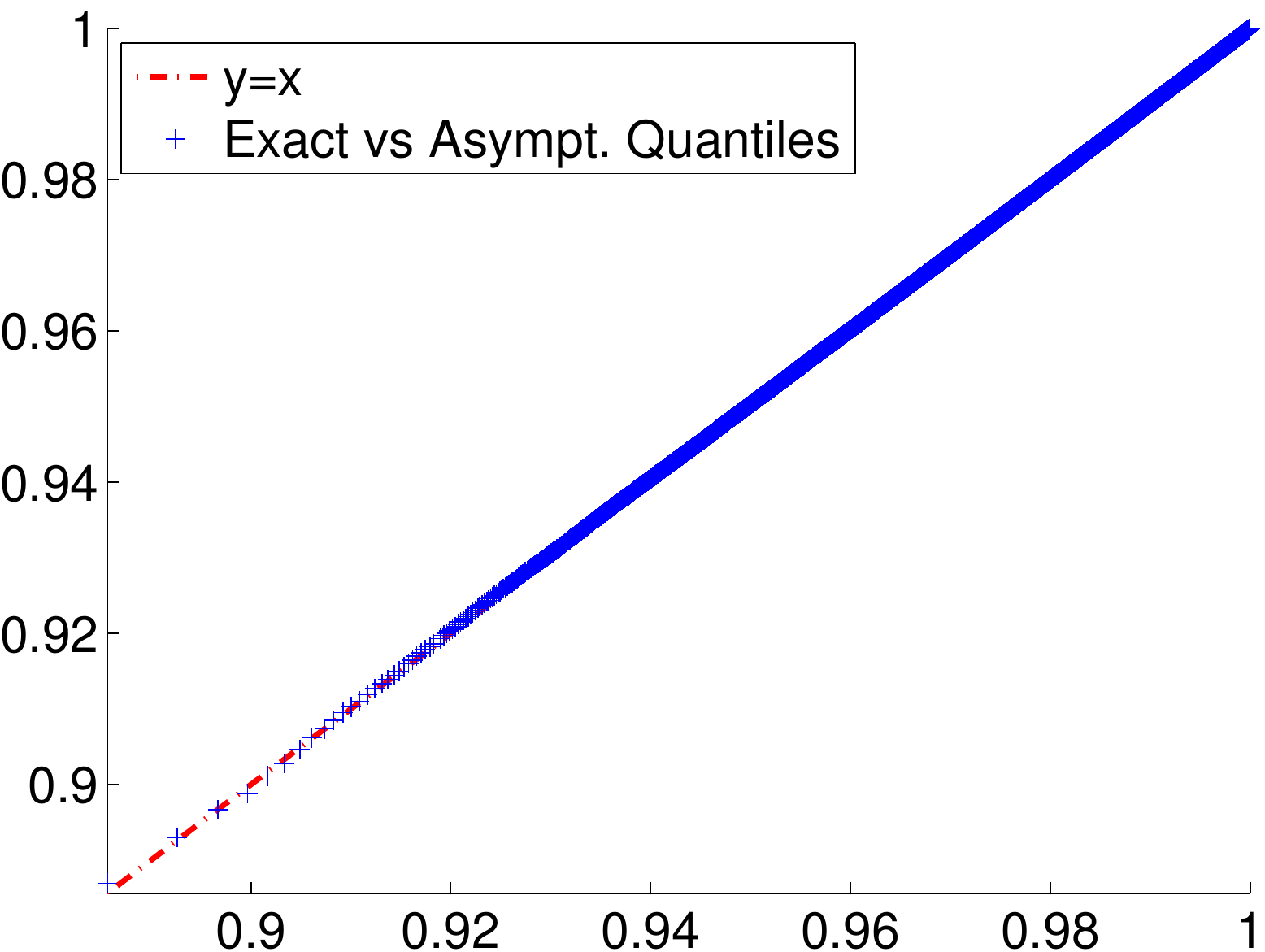}}
\caption{Comparisons of the exact and asymptotic projected distributions for the
$n=10$ steps vMF random
walk on ${\mathbb S}^2$ (i.e. $p=3$) with
concentration parameter $\kappa$.
Top row: $\kappa=10$, middle row: $\kappa=100$, bottom row: $\kappa=1000$.
Left column: {\em pdf}s  for the exact (blue cross) and asymptotic (green circles) distributions.
Right column: empirical qq-plots for the exact distribution quantiles vs the asymptotic
ones ($10^7$ samples).}
\label{fig:asymptconv}

\end{figure}

\subsection{VMF multiple scattering process}

The number of scattering events $N(t)$, which corresponds to the number of steps in the vMF random walk, 
is now assumed to be distributed as a  Cox process with mixing variable $\Lambda(t)$:
$N(t)| \Lambda(t) \sim \mathcal{P}(\Lambda_t)$.

\subsubsection{Asymptotic approximation of the vMF  multiple scattering process}
Based on the vMF asymptotic distribution for the vMF random walk given in Thm. \ref{thm:vMFapprox},  
it is possible to obtain an asymptotic distribution for the vMF scattering process
\begin{prop} \label{prop:asymptScat}
Assume that there exists $a>0$ such that $E[\Lambda(t)^a]< + \infty$.
In the large $\kappa$ case, 
an asymptotic expression for the {\em pdf} of the direction $\bx_t$ given in \eqref{eq:pdfexpScat} 
for the vMF multiple scattering process is given by the following linear mixture
\begin{align}
  \label{eq:AsympScatPdf}
  \tilde{f}(\bx_t ;\bmu, \kappa) &=  {\cal P}_0\left[f_{\Lambda_t}\right] \delta_{\bmu}(\bx_t) + 
  \sum_{n \ge 1}  {\cal P}_n\left[f_{\Lambda_t}\right]  f(\bx_t ; \bmu, \tilde{\kappa}_n ),
\end{align}
where $f(\bx_t ; \bmu, \tilde{\kappa}_n )$ is the {\em pdf}  of the vMF distribution
$M_p(\bmu,\tilde{\kappa}_n)$ given in Thm. \ref{thm:vMFapprox}.\\

Moreover, this yields the following approximation of the Fourier coefficients:
\begin{align*}
 \widehat{h^{\otimes>0}}_{\ell}= \widetilde{h}^{\otimes>0}_{\ell} +  O \left( \left(\frac{1}{\kappa} \right)^{\frac{3a}{a+3}}  \right),
\end{align*}
where $\widehat{h^{\otimes>0}}_{\ell}$ is the $\ell$th-order Fourier coefficient of the continuous part  of the
$\bx_t$ distribution given in \eqref{eq:FourierDiffPos},
while $\widetilde{h}^{\otimes>0}_{\ell}$ denotes the $\ell$th-order Fourier coefficient of the  continuous part of the 
asymptotic distribution.
\end{prop}
\begin{proof}
 The asymptotic pdf is obtained by plugging the vMF asymptotic distributions given in Thm. \ref{thm:vMFapprox}  in the 
 multiple scattering mixture {\em pdf}  given in \eqref{eq:ftimet}. To show that the distributions are 
 asymptotically equivalent, it is sufficient to show that their Fourier coefficients are asymptotically equivalent.
  
The expression of the $\ell$th order Fourier coefficient of the continuous part of the density 
 of $\bx_t$ given in \eqref{eq:FourierDiffPos} can be splitted in two terms:
 \begin{align}
  \widehat{h^{\otimes>0}}_{\ell} &=  \sum_{n = 1}^{m_{\kappa}} {\cal P}_n\left[f_{\Lambda_t}  \right] \widehat{f}^{\otimes n}_{\ell} + 
  \sum_{n > m_{\kappa}} {\cal P}_n\left[f_{\Lambda_t}  \right] \widehat{f}^{\otimes n}_{\ell} ,
 \label{eq:sumsplit}
 \end{align}

According to Thm \ref{thm:vMFapprox}, 
$\widehat{f}^{\otimes n}_{\ell}= \widetilde{f}^n_{\ell} +  O \left( \left(\frac{m_{\kappa}}{\kappa}\right)^3 \right),$ 
for all $ n \le m_{\kappa}$, 
thus the finite sum in the right-hand side of \eqref{eq:sumsplit} can be dominated as 
\begin{align}
\label{eq:firstRHS}
 \sum_{n = 1}^{m_{\kappa}} {\cal P}_n\left[f_{\Lambda_t}  \right] \widehat{f}^{\otimes n}_{\ell} &= 
 \sum_{n = 1}^{m_{\kappa}} {\cal P}_n\left[f_{\Lambda_t}  \right] \widetilde{f}^n_{\ell} + \sum_{n = 1}^{m_{\kappa}} {\cal P}_n\left[f_{\Lambda_t} \right] \times  O \left( \left(\frac{m_{\kappa}}{\kappa}\right)^3 \right), \nonumber \\
 &= \sum_{n = 1}^{m_{\kappa}} {\cal P}_n\left[f_{\Lambda_t}  \right] \widetilde{f}^n_{\ell} + O\left( \left(\frac{m_{\kappa}}{\kappa}\right)^3 \right),
\end{align}
since $\sum_{n = 1}^{m_{\kappa}} {\cal P}_n\left[f_{\Lambda_t} \right] \le \sum_{n \ge 1} {\cal P}_n\left[f_{\Lambda_t} \right] \le 1$.

The second series in the right-hand side of \eqref{eq:sumsplit} can be expressed as
\begin{align}
\label{eq:secondRHS}
\sum_{n > m_{\kappa}} {\cal P}_n\left[f_{\Lambda_t}  \right] \widehat{f}^{\otimes n}_{\ell}&= 
\sum_{n > m_{\kappa}} {\cal P}_n\left[f_{\Lambda_t}  \right] \widetilde{f}^{n}_{\ell} + 
\sum_{n > m_{\kappa}} {\cal P}_n\left[f_{\Lambda_t}  \right] \left( \widehat{f}^{\otimes n}_{\ell} - \widetilde{f}^{n}_{\ell} \right) ,
\end{align}
with
\begin{align*}
\left|\sum_{n > m_{\kappa}} {\cal P}_n\left[f_{\Lambda_t}  \right] \left( \widehat{f}^{\otimes n}_{\ell} - \widetilde{f}^{n}_{\ell} \right)  \right| \le  \sum_{n > m_{\kappa}} {\cal P}_n\left[f_{\Lambda_t}  \right] \left( \left| \widehat{f}^{\otimes n}_{\ell}\right| + \left| \widetilde{f}^{n}_{\ell} \right| \right)
 \le 2 \sum_{n > m_{\kappa}} {\cal P}_n\left[f_{\Lambda_t}  \right]= 2 \Pr{ \left( N_t > m_{\kappa} \right) }, 
\end{align*}
where $N_t$ is the Cox process counting the scattering events whose mixing intensity variable is $\Lambda_t$. A classical result about mixed Poisson distribution 
\cite{karlis2005}
yields that
\begin{align*}
 \Pr{ \left( N_t > m_{\kappa} \right) } &= \int_0^{+\infty} e^{-\lambda_t} \frac{\lambda_t^{m_{\kappa}}}{ m_{\kappa}! } 
 \Pr{ \left( \Lambda_t > \lambda_t  \right) } d\lambda_t.
\end{align*}
Because 
$\Lambda_t$ is a positive random variable, the Markov inequality ensures that for $a>0$
\begin{align*}
 \Pr{ \left( \Lambda_t > \lambda_t  \right) }& \le \frac{
E[\Lambda_t^a]}{\lambda_t^a}, \qquad \textrm{ for all } \lambda_t>0.
\end{align*}
As a consequence, 
\begin{align*}
\Pr{ \left( N_t > m_{\kappa} \right) } & \le E[\Lambda_t^a] \int_0^{+\infty} e^{-\lambda_t} \frac{\lambda_t^{m_{\kappa}-a}}{ m_{\kappa}! } d\lambda_t 
= E[\Lambda_t^a] \frac{\Gamma\left( m_{\kappa}-a +1 \right)}{\Gamma\left( m_{\kappa} +1 \right)},
\end{align*}
thus for large $m_{\kappa}$,  $\Pr{ \left( N_t > m_{\kappa} \right) } = O\left( \left(\frac{1}{m_{\kappa}}\right)^a\right)$, according to both the Stirling formula and the assumption that 
$E[\Lambda_t^a]<+\infty$. Eq. \eqref{eq:secondRHS} can then be rewritten as
\begin{align}
\label{eq:secondRHSbis}
\sum_{n > m_{\kappa}} {\cal P}_n\left[f_{\Lambda_t}  \right] \widehat{f}^{\otimes n}_{\ell}&= 
\sum_{n > m_{\kappa}} {\cal P}_n\left[f_{\Lambda_t}  \right] \widetilde{f}^{n}_{\ell} + O\left( \left(\frac{1}{m_{\kappa}}\right)^a\right). 
\end{align}
Plugging now \eqref{eq:firstRHS} and \eqref{eq:secondRHSbis} in \eqref{eq:sumsplit} yields that
\begin{align*}
  \widehat{h^{\otimes>0}}_{\ell} &= \widetilde{h}^{\otimes >0}_{\ell}
  + O\left( \left(\frac{m_{\kappa}}{\kappa}\right)^3\right) +
  O\left( \left(\frac{1}{m_{\kappa}}\right)^a\right), 
\end{align*}  
where $\widetilde{h}^{\otimes >0}_{\ell} = \sum_{n \ge 1}  {\cal P}_n \left[f_{\Lambda_t}  \right]  \widetilde{f}_{\ell}^n$ 
is the $\ell$th order Fourier coefficient of $\sum_{n \ge 1}  {\cal P}_n\left[f_{\Lambda_t}\right]  f(\bx_t ; \bmu, \tilde{\kappa}_n )$
due to the orthogonality of the Legendre polynomials.
Finally, setting $m_{\kappa}= \lfloor \kappa^{\gamma} \rfloor$ with $\gamma=\frac{3}{3+a} \in (0,1)$ gives the expected result. This concludes the proof.
\end{proof}
Note that when all the moments of the intensity variable $\Lambda_t$ exist - this is the case for the 
Poisson process ($\Lambda_t$ is deterministic) or the negative binomial process ($\Lambda_t$ obeys a Gamma distribution) - the 
asymptotic approximation given in  Prop. \ref{prop:asymptScat} yields an almost third order approximation of the Fourier coefficient for high
concentration parameter $\kappa$.

From a practical point of view, it is also interesting  to note that the asymptotic {\em pdf } approximation
given in Prop. \ref{prop:asymptScat} 
may give a numerically simpler way to evaluate the {\em pdf} in the vMF case than the Fourier expansion 
\eqref{eq:pdfexpScat}. In fact, this mixture expression gives a series expansion with positive terms and weights. 
In addition, the mixture model offers a very simple way to draw some random variables 
asymptotically distributed according to the vMF multiple scattering process for large enough concentration parameter $\kappa$.

\subsubsection{Estimation bounds}

From the estimation theory  perspective, it is of special interest to quantify 
the amount of information that the observed process carries about its distribution. 
In a parametric framework, the information about 
the parameters that govern the distribution is measured by the Fisher information matrix. 
Inverting the Fisher information  
provides now the Cramer-Rao lower bound (CRLB). This is 
a lower bound on the variance of any unbiased estimators of
the parameters to be estimated. In addition,  standard maximum likelihood 
estimators (MLEs) are known to be asymptotically unbiased and efficient under 
mild regularity conditions. This means that their large sample asymptotic variances approximately equal 
the CRLB.

The {\em pdf} expression of the multiple scattering process given in Prop. \ref{prop:FourierScat}, 
makes now possible to compute the Fisher information, and therefore
the CRLB. However the Fisher information matrix requires to determine the covariance 
of the first-order derivatives with respect to the process parameters 
(or equivalently the negative expectations of second-order derivatives)
of the log {\em pdf}. Closed-form expressions for both the Fisher information or the CRLB 
are difficult to obtain as this log {\em pdf} has no simple tractable
expression. 
In such situation, it is very usual to approximate the expectations 
by using Monte Carlo methods. This allows one to numerically evaluate the 
CRLB based on the {\em pdf} Fourier expansion \eqref{eq:pdfexpScat} for the vMF 
multiple scattering process. Moreover, in the 
high concentration case, i.e. for large $\kappa$, the vMF
mixture representation given in Prop. \ref{prop:asymptScat} yields another 
simple way to approximate the CRLB with Monte Carlo methods.

The vMF concentration parameter $\kappa>0$ and the mean
resultant length $\rho \in (0,1)$ of the vMF random steps 
are related by a one to one transformation $\rho=A_p(\kappa)$ given in \eqref{eq:rhoVMF}. 
Therefore the random steps are reparametrized in the remainder by the scalar 
$\rho$. This yields a simple interpretation:
the closer is $\rho$ to $1$, the more concentrated is the distribution, the closer 
to $0$, the more uniform distribution.
The CRLB for the parameter $\rho$ reduces to 
$\frac{ \left[A_p'(\kappa) \right]^2} {I(\kappa)}$ where $I(\kappa)$ is the 
Fisher information for $\kappa$.

When the number $N(t)$ of scattering events is a Poisson process, 
the intensity variable of the compound Cox process $\Lambda_t$ is a deterministic
value $\Lambda_t=\lambda_t$.
The distribution of the multiple scattering process $\bx_t$ is then parametrized by the
vector $\left( \rho,\lambda_t \right)$. Figs \ref{fig:PoissRhovsKappa} and  \ref{fig:PoissLvsKappa}
depict the CRLBs for the parameters  $\rho$ 
and $\lambda_t$  respectively as a function of 
$\kappa=A_p^{-1}(\rho)$ when the dimension is $p=3$. The value of the Poisson
intensity is set to $\lambda_t=10$. One can see that the CRLBs quickly increase 
when the concentration decreases. In fact, in the limit case where 
$\rho$ tends to zero (or equivalently, $\kappa$ tends to zero), the distribution of 
$\bx_t$ converges to an uniform distribution  on ${\mathcal S}^{p-1}$ and the model is not  
identifiable anymore.
For larger values of $\rho$ (i.e. for high $\kappa$), good estimation performances
can be reached. Moreover there is little gain possible for the parameter
$\lambda_t$ that governs the number $N(t)$ of scattering events
when $\rho$ continues to 
converge to $1$. Finally, one can see that the CRLBs computed for the 
high  concentration asymptotic distribution of $\bx_t$ given in Prop. \ref{prop:asymptScat}
are in good agreement with the exact one when $\rho$ is close to $1$.

\begin{figure}[htbp!]
\centering
\subfigure[CRLBs for $\rho$\label{fig:PoissRhovsKappa}] {
        \includegraphics[width=.45\linewidth]{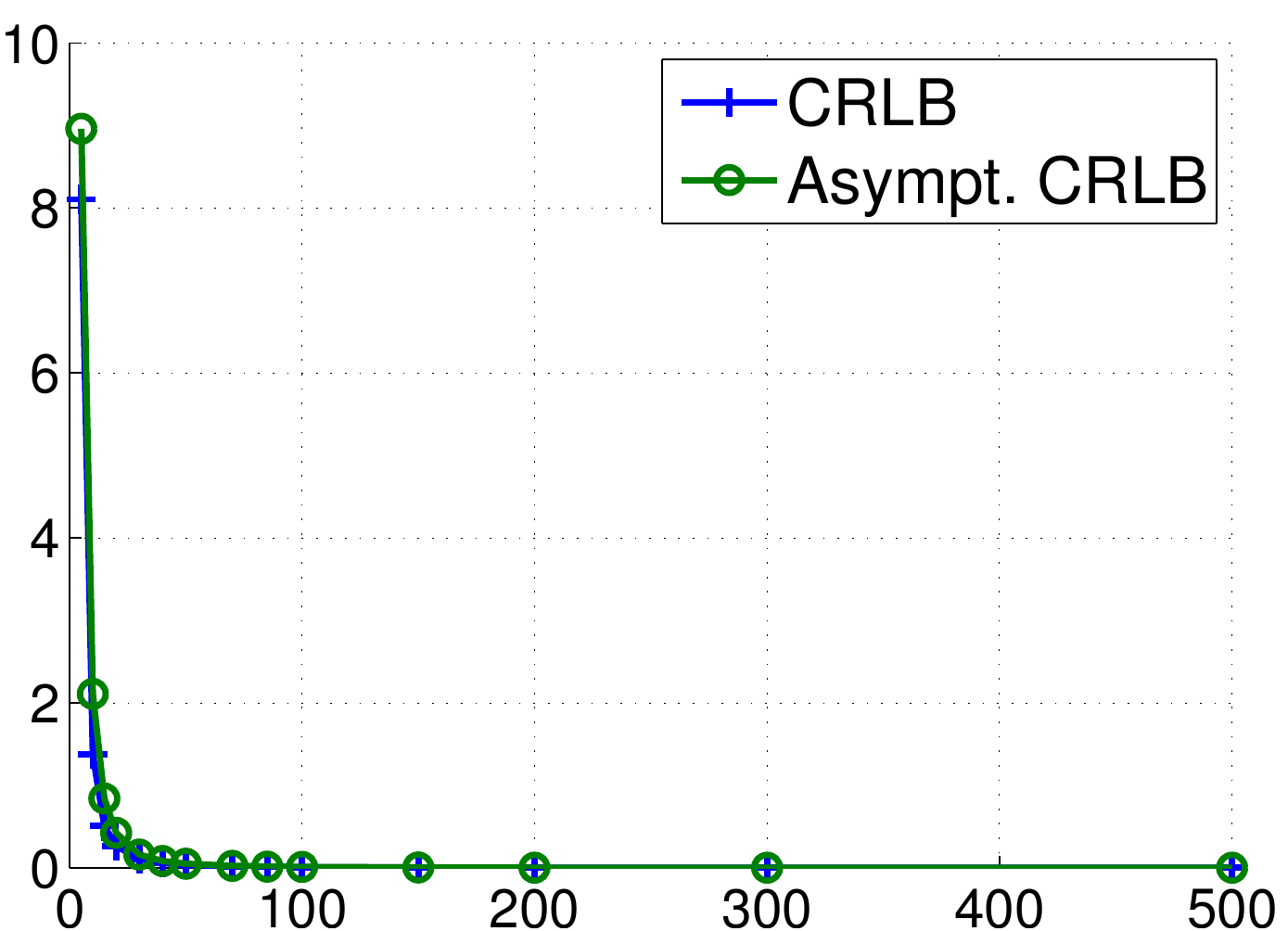}
}
\subfigure[CRLBs for $\lambda_t$\label{fig:PoissLvsKappa}]{
        \includegraphics[width=.48\linewidth]{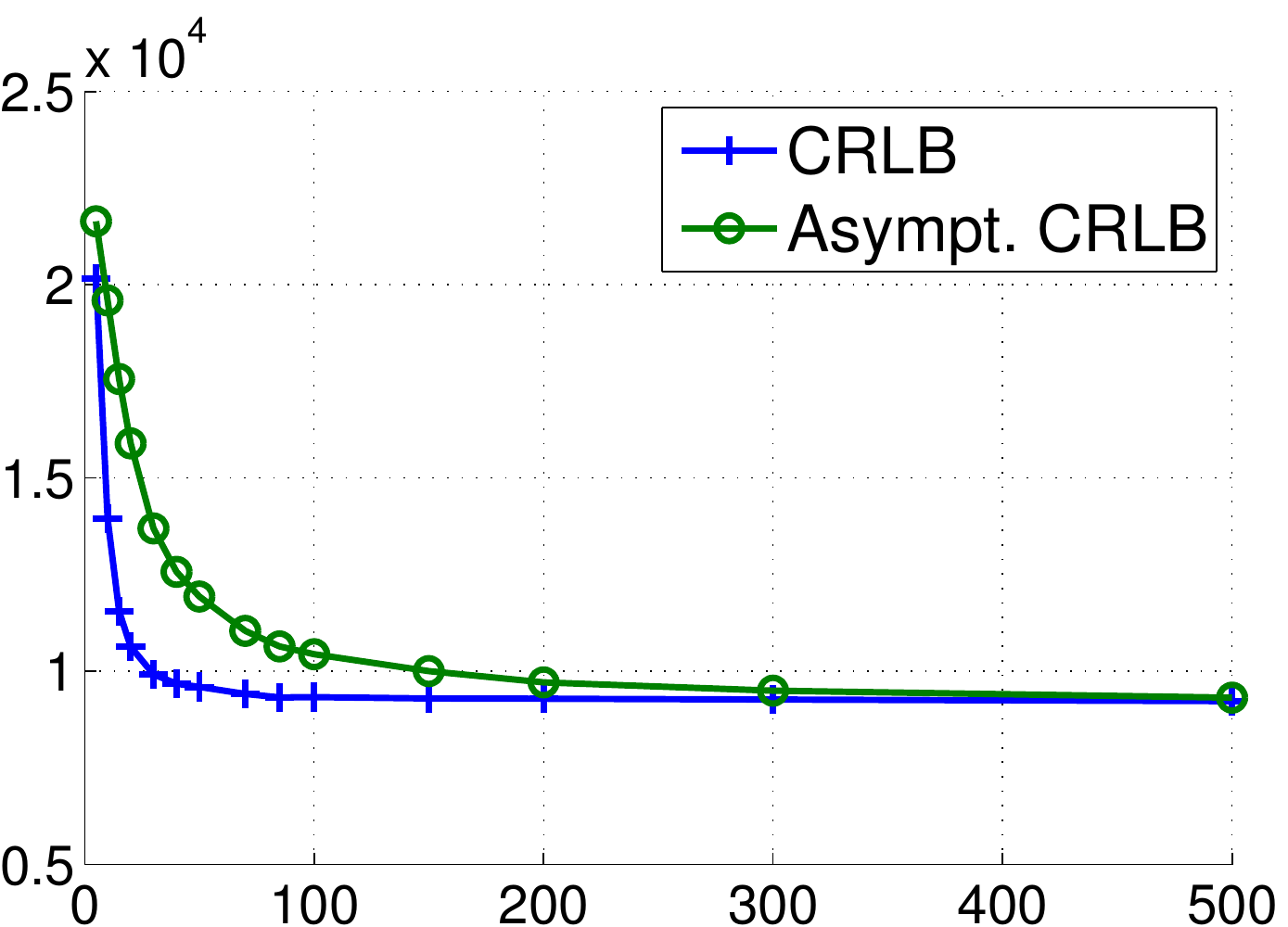}
}
\caption{CRLBs vs $\kappa=A_p^{-1}(\rho)$ for the Poisson multiple scattering
process ($\lambda_t=10$, $p=3$).
Blue curve:  exact distribution of $\bx_t$. Green curve: high concentration 
asymptotic distribution of $\bx_t$.)
\label{fig:PoissCRLBvsKappa}}
\end{figure}

Similar results are shown in Fig. \ref{fig:PoissCRLBvsL} when
the mean resultant length is set to a fixed value
$\rho= 0.99$ (i.e. $\kappa \approx 100$) while the intensity $\lambda_t$ 
of the Poisson process varies. Fig. \ref{fig:PoissRhovsL} shows that 
for too high $\lambda_t$, the model becomes 
hardly identifiable (in this case the distribution of $\bx_t$ converges to 
an uniform distribution). When $\rho$ is fixed and for small enough $\lambda_t$, 
the high  concentration asymptotic CRLB  is in good agreement with the exact one.

\begin{figure}[htbp!]
\centering
\subfigure[CRLBs for $\rho$\label{fig:PoissRhovsL}] {
        \includegraphics[width=.45\linewidth]{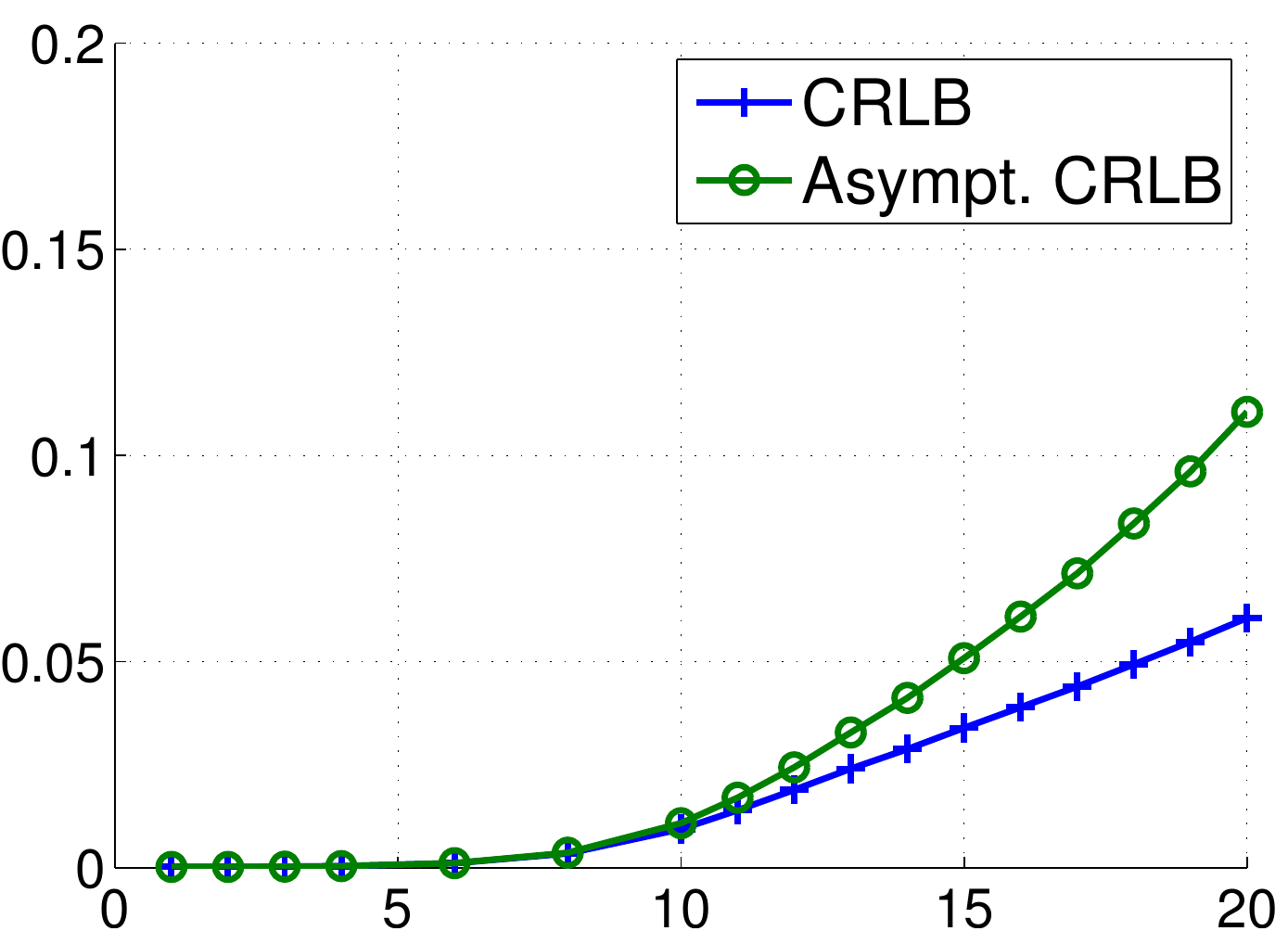}
}
\subfigure[CRLBs for $\lambda_t$\label{fig:PoissLvsL}]{
        \includegraphics[width=.48\linewidth]{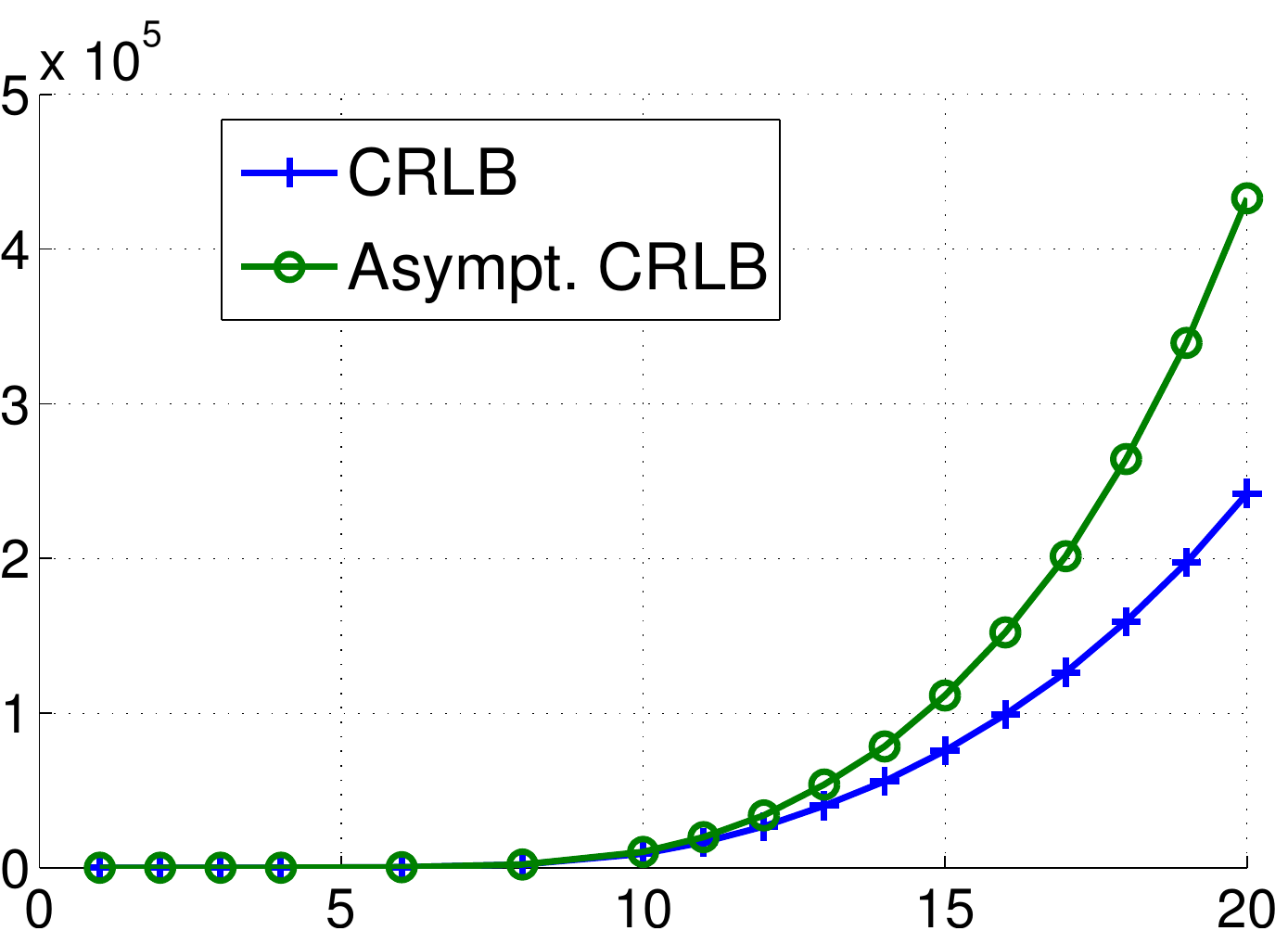}
}
\caption{CRLBs vs $\lambda_t$ for the Poisson multiple scattering
process ($\rho=0.99$).
Blue curve:  exact distribution of $\bx_t$. Green curve: high concentration 
asymptotic distribution of $\bx_t$.)
\label{fig:PoissCRLBvsL}}
\end{figure}

To conclude, we consider the more general case where 
the intensity is distributed as a Gamma process $\Lambda_t \sim \mathcal{G}(\xi_t,\theta)$,  
thus $N(t)$ is a Negative Binomial process paramaterized by the vector 
$(\rho,\theta,\xi_t)$. Figs \ref{fig:NBRhovsKappa}), \ref{fig:NBThetavsKappa} and
\ref{fig:NBXiLvsKappa} depict the CRLBs for the parameters  $\rho$, $\theta$, 
and $\xi_t$ respectively as a function of $\kappa=A_p^{-1}(\rho)$ when the dimension is $p=3$. 
The parameters are set to $\theta=1$ and $\xi_t=10$.
Similar conclusoins to the one reported 
for Fig. \ref{fig:PoissCRLBvsKappa} can be drawn. 

The CRLBs presented for both Poisson and Negative Binomial cases outline the accuracy of the multiple scattering process model when the number of scattering events is low, which indicates its appropriateness and usefulness in forward scattering regimes occuring well before the full diffusion regime.

\begin{figure}[htbp!]
\centering
\subfigure[CRLBs for $\rho$\label{fig:NBRhovsKappa}] {
        \includegraphics[width=.45\linewidth]{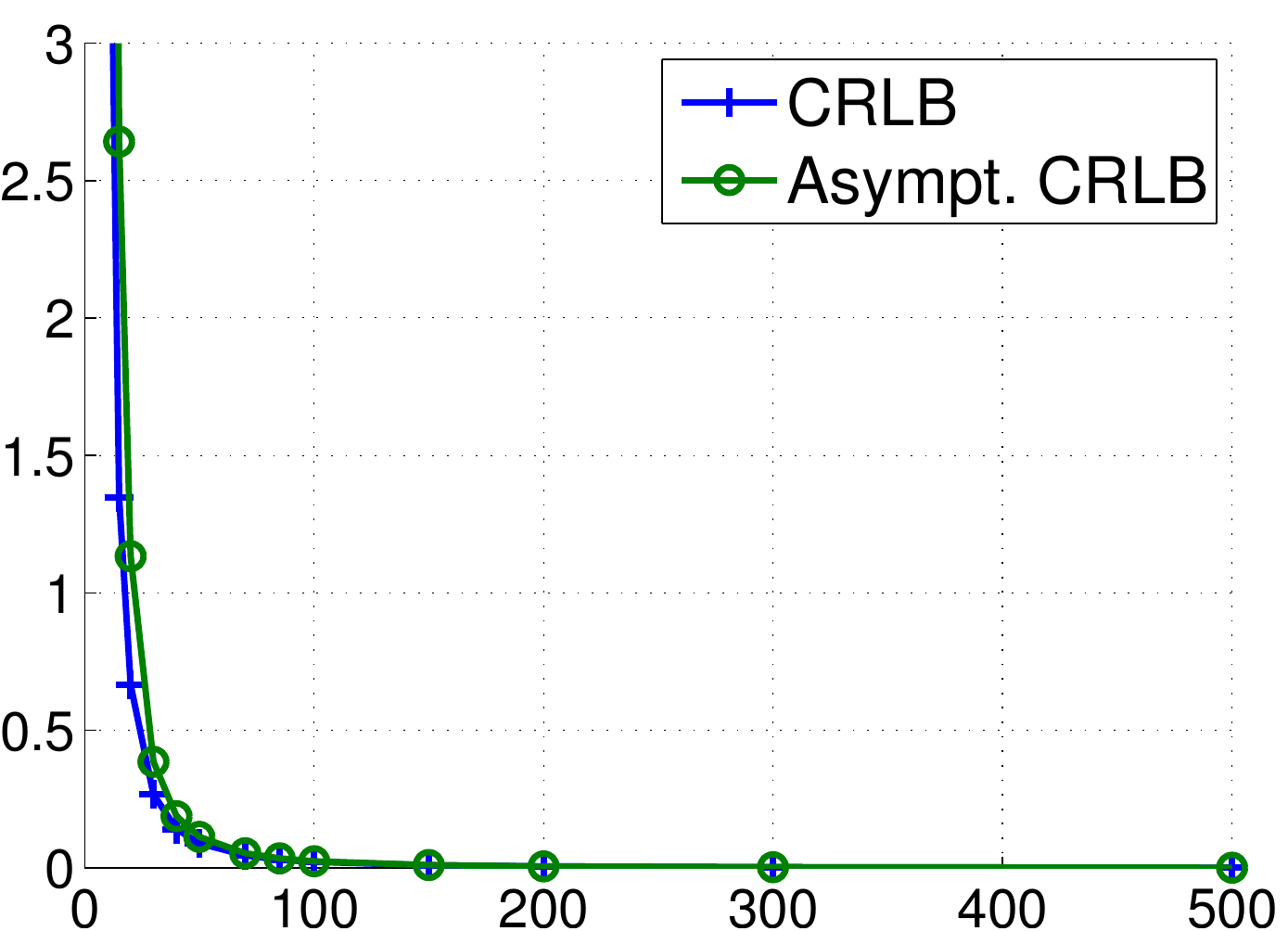}
}
\subfigure[CRLBs for $\theta$\label{fig:NBThetavsKappa}]{
        \includegraphics[width=.48\linewidth]{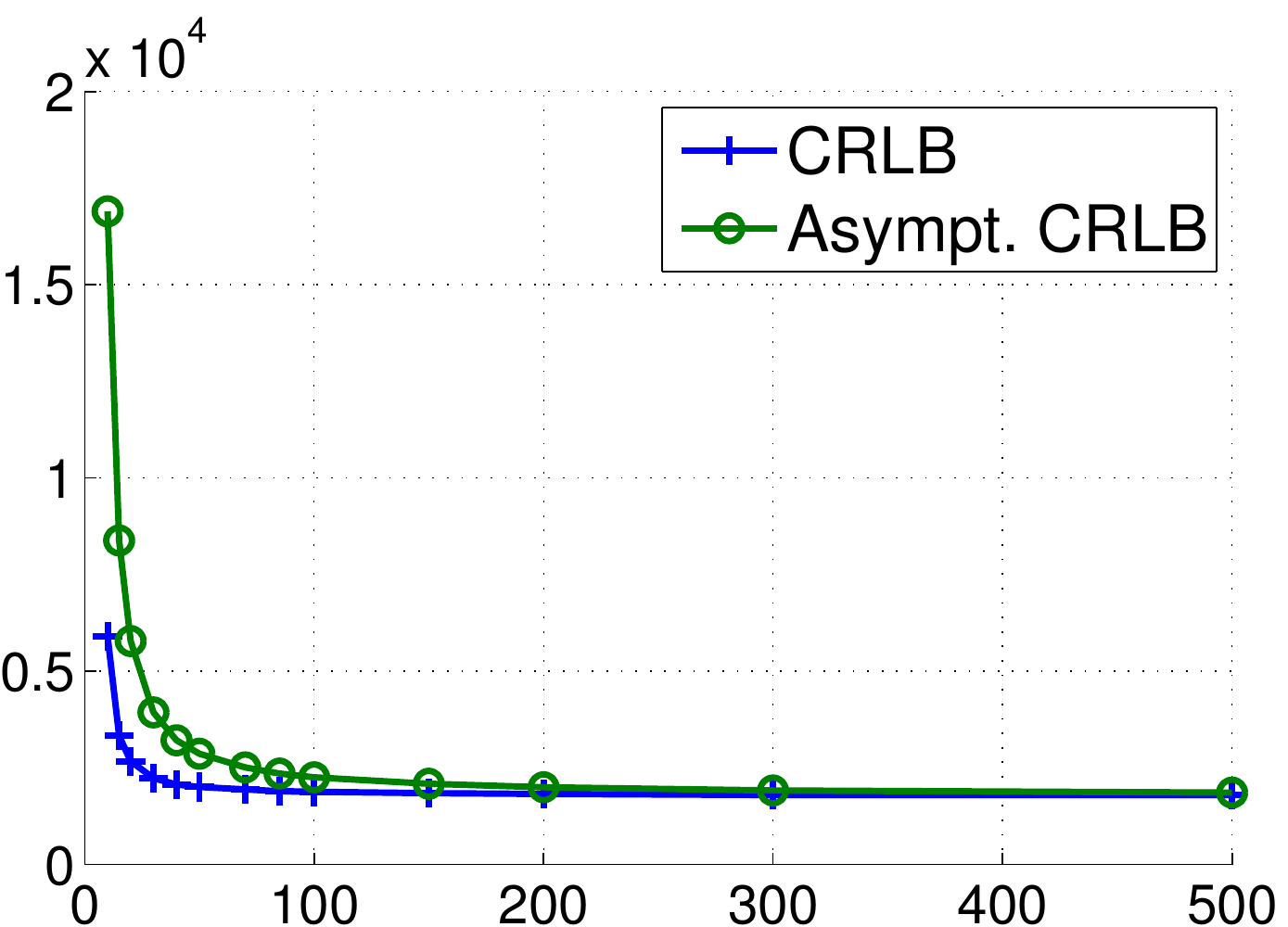}
}
\subfigure[CRLBs for $\xi_t$\label{fig:NBXiLvsKappa}]{
        \includegraphics[width=.48\linewidth]{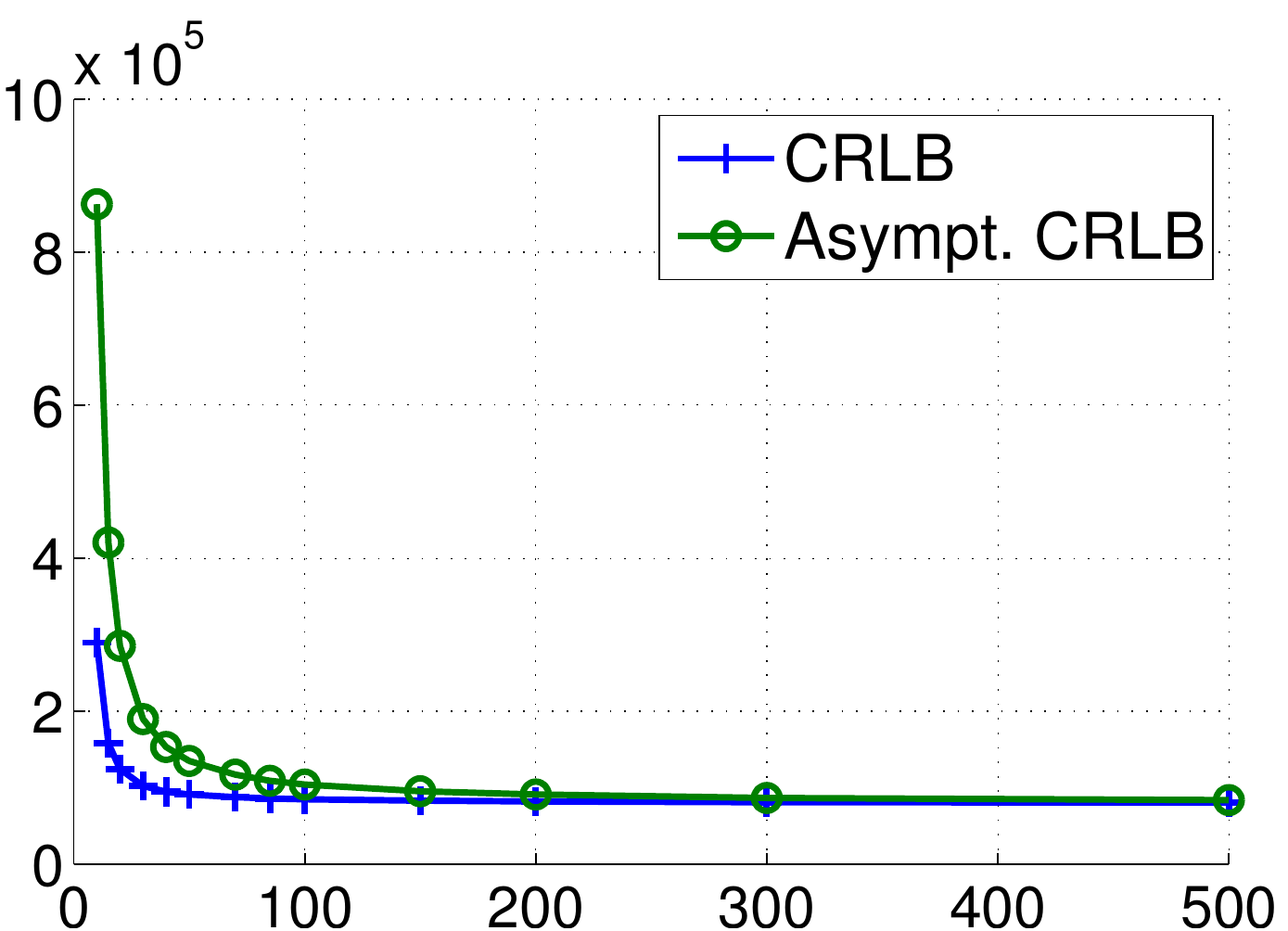}
}
\caption{CRLBs vs $\kappa=A_p^{-1}(\rho)$ for the Negative Binomial 
multiple scattering process ($\theta=1$, $\xi_t=10$, $p=3$).
Blue curve:  exact distribution of $\bx_t$. Green curve: high concentration 
asymptotic distribution of $\bx_t$.
\label{fig:NBCRLBvsKappa}}
\end{figure}

\section{Conclusion}
In this paper, we have studied multiple isotropic random walks and scattering processes on hyperspheres and obtained Fourier expansions for their {\em pdf}s. The case where the random steps follow a von Mises Fisher law on ${\mathbb S}^{p-1}$ has been detailed and asymptotic approximation for multi-convoltion of such densities have been introduced. The obtained expressions allow to numerically compute lower estimation bounds for the parameters of multiple scattering processes. These bounds should be of interest for future studies of estimation techniques (Method of moments, Bayesian, etc.) for such processes. The abundance of multiple scattering situations in engineering applications should provide applications for the presented results.      

\section*{Acknowledgments}
The authors are grateful to Pr. Peter Jupp for fruitful discussions
regarding directional statistics and convolution on double coset 
spaces.
\appendix
\subsection{Proof of the unimodality of the convolution of unimodal and rotationally symmetric distributions}
\label{app:uni}
Let $f_1$ and $f_2$ be the {\it pdf}s of two absolutely continuous  unimodal and rotationally symmetric distributions on
${\mathbb S}^{p-1}$ with the same mode $\bmu \in {\mathbb S}^{p-1}$.  These  {\it pdf}s express as continuous functions of the only
cosine $\bx^T \bmu$, that is $f_1(\bx;\bmu) = g_1(\bx^T \bmu)$ and $f_2(\bx;\bmu) = g_2(\bx^T \bmu)$. Moreover, due
to the unimodality property, $g_1$ and $g_2$ are increasing functions from $[-1,1]$ to $\mathbb{R}^+$.
As explained in \ref{sec:coset}, $f_1$ and $f_2$  belong
to the double coset space $L^1\left( SO(p-1)\backslash SO(p)/ SO(p-1), \mathbb{R} \right)$ where $SO(p-1)$ stands for
the rotation subgroup such that the axis defined by the unit vector $\bmu$ is left invariant.
This space is stable by convolution, and
the resulting convolved  {\it pdf} is rotationally symmetric about $\bmu$ and reads
\begin{align*}
   f(\bx)&=
   \left( f_2 \star_{\bmu} f_1 \right) ({\bx}) =
   \left(f_2 \star f_1\right) ({\bx}^T\bmu) \\
& = \int_{{\mathbb S}^{p-1}} g_2({\bx}^T{\bx_1})g_1({\bx_1}^T\bmu) d\bx_1.
\end{align*}
We want to show now that the convolved distribution
is unimodal with mode $\bmu$.

The sketch of the proof is inspired by the proof given in \cite{Purk98} of the equivalent
property on the real line.
Note first that the convolved {\it pdf} can be expressed as the following expectation:
\begin{align}
   f(\bx)
& = \int_{{\mathbb S}^{p-1}} g_2({\bx}^T{\bx_1})g_1({\bx_1}^T\bmu) d\bx_1,\nonumber \\
& = E \left[ g_2(\bX_1^T\bx)\right], \label{eq:uniesp}
\end{align}
where $\bX_1$ is a random vector in ${\mathbb S}^{p-1}$ with  {\it pdf} $f_1$.
Since  $g_2(\bX_1^T \bx)$ is a continuous positive variable, one gets
\begin{align}
E \left[ g_2(\bX_1^T\bx)\right]& =
\int_0^{+\infty} \Pr ( g_2(\bX_1^T\bx) \ge u ) du,\\\label{eq:unipos}
\end{align}
Furthermore,
since $g_2$ is continuous and increasing on $(-1,1)$, for all $u$ in the image of $g_2$
\begin{align*}
\Pr\left( g_2\left(\bX_1^T \bx \right) \ge u  \right) &=
\Pr\left(\bX_1^T \bx \ge \delta  \right),\\
&= \int_{\bx_1 \in {\mathbb S}^{p-1} \cap  \bx^T \bx_1 \ge \delta} g_1 \left(\bmu^T \bx_1\right) d\bx_1.
\end{align*}
where $\delta= g_2^{-1}(u) \in (-1,1)$. In the remainder, since
we only consider unit vectors, the hypersphere constraint $\bx_1 \in {\mathbb S}^{p-1}$ on the integration domain
will be omitted to simplify the notation.

The normal-tangent decomposition along the axis $\bmu$ yields now that
$\bx = t\bmu + \sqrt{1-t^2} \bxi$ where
$t \in [-1,1]$ and $\bxi \in \bmu^{\bot} \cap S^{p-1}$.

\begin{lemma}
 \label{lemma1}
For any fixed  $\bxi \in \bmu^{\bot} \cap S^{p-1}$ and $\delta \in (-1,1)$,
the function $h$ defined as
\begin{align*}
   t \in [-1,1] \mapsto h(t)= \int_{\bx^T_t \bx_1 \ge \delta} g_1\left( \bmu^T \bx_1 \right)  d\bx_1,
\end{align*}
where $\bx_t= t\bmu + \sqrt{1-t^2} \bxi$ is increasing.
\end{lemma}
\begin{proof}
 See section \ref{sec:lemma1proof}.
\end{proof}
According to Lemma \ref{lemma1}, if $\bx' = t'\bmu + \sqrt{1-{t'}^2} \bxi$ with
$-1 \le t < t' \le 1$, then $\Pr\left(\bX_1^T \bx \ge \delta  \right) < \Pr\left(\bX_1^T \bx' \ge \delta  \right)$ for all
$\delta \in (-1,1)$.
This yields directly that $f(\bx)< f(\bx')$ according to \eqref{eq:unipos} and \eqref{eq:uniesp}.
Due to the rotational symmetry about $\bmu$, the convolved density expresses
as $f(\bx)= g(\bmu^T \bx)$ and depends on the only tangent part $t=\bmu^T \bx$. Thus
the inequality $f(\bx)< f(\bx')$  extends to the case where
$\bx' = t'\bmu + \sqrt{1-{t'}^2} \bxi'$ for any $\bxi'\in \bmu^{\bot} \cap S^{p-1}$.
Finally, we obtain that $g$ is an increasing function on $[-1,1]$: for all $\bx,\bx' \in S^{p-1}$,
$f(\bx)= g(\bmu^T \bx)< f(\bx')= g(\bmu^T \bx')$ iff $\bmu^T \bx < \bmu^T \bx'$,
and the maximum is reached for $\bmu^T \bx'=1$, i.e. when $\bx'=\bmu$.
This concludes the proof.
\subsection{Proof of Lemma \ref{lemma1}}
\label{sec:lemma1proof}
Consider the following unit vectors,
\begin{align*}
\bx_t& = t\;\bmu + \sqrt{1-{t\,}^2} \bxi,\\
\bx_{t'}&=t'\bmu + \sqrt{1-{t'}^2} \bxi,
\end{align*}
with $-1 \le t <t'\le 1$ and $\bxi \in \bmu^{\bot} \cap S^{p-1}$.
Introduce the reflection matrix $R\in O(p)$ across
the axis directed by $\bx_t+\bx_{t'}$,  $R= \left[\frac{(\bx_t+\bx_{t'}) (\bx_t+\bx_{t'})^T }{1 +\bx_t^T \bx_t'}-I_{p}\right]$, where $I_p$ is the $p\times p$
identity matrix.

This interchanges $ \bx_{t}$ and $\bx_{t'}$, i.e. $R \bx_{t}= \bx_{t'}$ and $R \bx_{t'}= \bx_{t}$.
Let $\bu_1 \in {\mathbb S}^{p-1}$ be defined as $\bu_1 = R \bx_1$. Since the reflection matrix $R$ satisfies $R=R^T$, it comes that
\begin{align}
\begin{split}
 \bx_{t'}^T \bx_1 &=  (R \bx_{t})^T \bx_1= \bx_{t}^T R \bx_1 = \bx_{t}^T \bu_1,\\
 \bx_{t}^T \bx_1 &=  (R \bx_{t'})^T \bx_1= \bx_{t'}^T R \bx_1 = \bx_{t'}^T \bu_1.
  \end{split}
  \label{eq:decompDom}
\end{align}

Moreover one gets that
\begin{align*}
 h(t) =&
 \int_{\substack{\ \, \bx^T_t \bx_1 \ge \delta \\ \cap\, \bx^T_{t'} \bx_1 \ge \delta}} g_1\left( \bmu^T \bx_1 \right)  d\bx_1 +
 \underbrace{\int_{\substack{\ \, \bx^T_t \bx_1 \ge \delta \\ \cap\, \bx^T_{t'} \bx_1 < \delta}} g_1\left( \bmu^T \bx_1 \right)  d\bx_1}_{I}.
\end{align*}
The reflection matrix $R$ satisfies $R \bx_{t}= \bx_{t'}$,  $|\det{R}|=1$, $R^{-1}=R$ and \eqref{eq:decompDom}.
As a consequence, performing the substitution $\bu_1= R \bx_1$ in the integral denoted as $I$
yields
\begin{align}
I &= \int_{\substack{\ \, \bx^T_{t'} \bu_1 \ge \delta \\ \cap\, \bx^T_t \bu_1 < \delta}} g_1\left( \bmu^T R \bu_1 \right)  d\bu_1.
\label{eq:Ichange}
\end{align}

We need now to use the following result
\begin{lemma}
 \label{lemma2}
For all $\bx_1 \in {\mathbb S}^{p-1}$ such that $\bx_t^T \bx_1 > \bx_{t'}^T \bx_1$,
\begin{align}
 \bmu^T \bu_1 &> \bmu^T \bx_1.
 \label{eq:Ineqmu}
\end{align}
\end{lemma}
\begin{proof}[Proof of Lemma \ref{lemma2}]
When $\bx_t^T \bx_1 > \bx_{t'}^T \bx_1$, eq. \eqref{eq:decompDom} allows us to derive the following inequalities
\begin{align*}
 \bx_{t'}^T \bu_1 & >  \bx_{t'}^T \bx_1,\\
 \bx_{t}^T \bx_1 &  >  \bx_{t}^T \bu_1.
\end{align*}
Using the normal tangent decomposition, theses  inequalities express as
\begin{align}
  t' \bmu^T \bu_1 + \sqrt{1-{t'}^2} \bxi^T \bu_1 & > t' \bmu^T \bx_1 + \sqrt{1-{t'}^2} \bxi^T \bx_1,
  \label{eq:tp}\\
  t \  \bmu^T \bx_1 + \sqrt{1-{t\,}^2} \bxi^T \bx_1 & > t\; \bmu^T \bu_1 + \sqrt{1-{t\,}^2} \bxi^T \bu_1.
  \label{eq:t}
\end{align}
When $t'=1$, or $t=-1$ respectively, it comes directly from  \eqref{eq:tp}, or  \eqref{eq:t} respectively, that
$\bmu^T \bu_1 >  \bmu^T \bx_1$. We can thus assume that $-1<t <t' < 1$.
Multiplying both sides of inequality \eqref{eq:t} and \eqref{eq:tp}
by $\frac{1}{ \sqrt{1-t^2}}>0$ and $\frac{-1}{ \sqrt{1-{t'}^2}}<0$ respectively, and summing the resulting
inequalities yields
\begin{align*}
 \left( \alpha(t')-\alpha(t) \right) \bmu^T \bu_1 >
  \left(\alpha(t')-\alpha(t) \right) \bmu^T \bx_1,
\end{align*}
where $\alpha(z) = \frac{z}{ \sqrt{1-z^2}}$. The function $\alpha$ being increasing on $(-1,1)$,
the factor $\alpha(t')-\alpha(t)$ is positive since  $-1<t <t' < 1$. Thus the required inequality
holds.
\end{proof}
For all $\bu_1$ belonging to the integration domain defined in \eqref{eq:Ichange}, it comes that
$\bx_{t'}^T \bu_1 > \bx_{t}^T \bu_1$ which is equivalent to
$\bx_t^T \bx_1 > \bx_{t'}^T \bx_1$. Thus
$\bmu^T R \bu_1=  \bmu^T \bx_1 < \bmu^T  \bu_1$ according to Lemma \ref{lemma2}. Since $g_1$ is increasing,
it comes that
\begin{align*}
I < \int_{\substack{\ \, \bx^T_{t'} \bu_1 \ge \delta \\ \cap\, \bx^T_t \bu_1 < \delta}} g_1\left( \bmu^T \bu_1 \right)  d\bu_1.
\end{align*}
As a consequence,
\begin{align*}
 h(t) < &
 \int_{\substack{\ \, \bx^T_t \bx_1 \ge \delta \\ \cap\, \bx^T_{t'} \bx_1 \ge \delta}} g_1\left( \bmu^T \bx_1 \right)  d\bx_1 +
\int_{\substack{\ \, \bx^T_{t'} \bu_1 \ge \delta \\ \cap\, \bx^T_t \bu_1 < \delta}} g_1\left( \bmu^T \bu_1 \right)  d\bu_1,\\
& =  \int_{\substack{\ \, \bx^T_{t'} \bz \ge \delta \\ \cap\, \bx^T_{t} \bz \ge \delta}} g_1\left( \bmu^T \bz \right)  d\bz +
\int_{\substack{\ \, \bx^T_{t'} \bz \ge \delta \\ \cap\, \bx^T_{t} \bz < \delta}} g_1\left( \bmu^T \bz \right)  d\bz,\\
&  = \int_{\substack{\bx^T_{t'} \bz \ge \delta}} g_1\left( \bmu^T \bz \right)  d\bz=  h(t'),
\end{align*}
and the inequality holds for all $-1\le t< t' \le 1$.

\bibliographystyle{IEEEbib}
\bibliography{bibli}

\begin{thebibliography}{10}

\bibitem{Bhalerao}
A.~Bhalerao and C.-. Westin,
\newblock ``Hyperspherical von mises-fisher mixture (hvmf) modelling of high
  angular resolution diffusion mri,''
\newblock in {\em MICCAI 2007}, N.~Ayache, S.~Ourselin, and A.~Maeder, Eds.,
  vol. 4791 of {\em Lecture Notes in Computer Science}, pp. 236--243. Springer
  Berlin Heidelberg, 2007.

\bibitem{Bangert}
M.~Bangert, P.~Hennig, and U.~Oelfke,
\newblock ``Using an infinite von mises-fisher mixture to cluster treatment
  beam directions in external radiation therapy,''
\newblock in {\em Proceedings of the Ninth international conference on Machine
  Learning and applications}, 2010, pp. 746--751.

\bibitem{Tang}
H.~Tang and S.M. Chu,
\newblock ``Generative model-based speaker clustering via mixture of von
  mises-fisher distributions,''
\newblock in {\em Proceedings of the IEEE International Conference on
  Acoustics, Speech and Signal Processing (ICASSP)}, 2009, pp. 4101--4104.

\bibitem{Banerjee}
A.~Banerjee, I.S. Dhillon, J.~Ghosh, and S.~Sra,
\newblock ``Clustering on the unit hypersphere using von mises-fisher
  distributions,''
\newblock {\em Journal of Machine Learning Research}, vol. 6, pp. 1345 -- 1382,
  2005.

\bibitem{Ning}
X.~Ning, L.~Papiez, and G.~Sandison,
\newblock ``Compound-poisson-process method for the multiple scattering of
  charged particles,''
\newblock {\em Phys. Rev. E}, vol. 52, no. 5, pp. 5621--5633, Nov 1995.

\bibitem{LeBihanMargerin}
N.~Le Bihan and L.~Margerin,
\newblock ``Nonparametric estimation of the heterogeneity of a random medium
  using compound poisson process modeling of wave multiple scattering,''
\newblock {\em Physical Review E}, vol. 80, pp. 016601, 2009.

\bibitem{Said}
S.~Said, C.~Lageman, N.~{Le Bihan}, and J.H. Manton,
\newblock ``Decompounding on compact {L}ie groups,''
\newblock {\em IEEE Transactions on Information theory}, vol. 56, no. 6, pp.
  2766 -- 2777, 2010.

\bibitem{Perrin}
F.~Perrin,
\newblock {\em \'Etude math\'ematique du mouvement Brownien de rotation},
\newblock Ph.D. thesis, Facult\'e des sciences de Paris, 1925.

\bibitem{Dokmanic2010}
I.~Dokmani\'c and D.~Petrinovi\'c,
\newblock ``Convolution on the n -sphere with application to pdf modeling,''
\newblock {\em Signal Processing, IEEE Transactions on}, vol. 58, no. 3, pp.
  1157--1170, March 2010.

\bibitem{Franceschetti2004}
M.~Franceschetti, J.~Bruck, and L.J. Schulman,
\newblock ``A random walk model of wave propagation,''
\newblock {\em Antennas and Propagation, IEEE Transactions on}, vol. 52, no. 5,
  pp. 1304--1317, May 2004.

\bibitem{Durant2007}
S.~Durant, O.~Calvo-Perez, N.~Vukadinovic, and J.-J. Greffet,
\newblock ``Light scattering by a random distribution of particles embedded in
  absorbing media: full-wave {M}onte {C}arlo solutions of the extinction
  coefficient,''
\newblock {\em J. Opt. Soc. Am. A}, vol. 24, no. 9, pp. 2953--2962, September
  2007.

\bibitem{Jin2012}
C.~Jin, R.R. Nadakuditi, E.~Michielssen, and S.~Rand,
\newblock ``An iterative, backscatter-analysis based algorithm for increasing
  transmission through a highly-backscattering random medium,''
\newblock in {\em Statistical Signal Processing Workshop (SSP), 2012 IEEE}, Aug
  2012, pp. 97--100.

\bibitem{ghogho2001}
M.~Ghogho, O.~Besson, and A.~Swami,
\newblock ``Estimation of directions of arrival of multiple scattered
  sources,''
\newblock {\em Signal Processing, IEEE Transactions on}, vol. 49, no. 11, pp.
  2467--2480, Nov 2001.

\bibitem{Costa2014}
M.~Costa, V.~Koivunen, and H.V. Poor,
\newblock ``Estimating directional statistics using wavefield modeling and
  mixtures of von-mises distributions,''
\newblock {\em Signal Processing Letters, IEEE}, vol. 21, no. 12, pp.
  1496--1500, Dec 2014.

\bibitem{Mardia}
K.V. Mardia and P.E. Jupp,
\newblock {\em Directional statistics},
\newblock John Wiley \& Sons Ltd, 2000.

\bibitem{Kent78}
J.~Kent,
\newblock ``Limiting behaviour of the von {M}ises-{F}isher distribution,''
\newblock {\em Math. Proc. Camb. Phil. Soc.}, vol. 84, pp. 531--536, 1978.

\bibitem{Volker2013}
M.~Volker,
\newblock {\em Lectures on constructive approximation. Fourier, Spline, and
  wavelet methods on the real line, the sphere and the ball},
\newblock Birhauser, 2013.

\bibitem{Abramowitz}
M.~Abramowitz and I.~A. Stegun,
\newblock {\em Handbook of Mathematical Functions with Formulas, Graphs, and
  Mathematical Tables},
\newblock Dover Publications, 1972.

\bibitem{Dym1972}
H.~Dym and H.P. McKean,
\newblock {\em Fourier series and integrals},
\newblock Academic Press, 1972.

\bibitem{Ishimaru1999}
A.~Ishimaru,
\newblock {\em Wave propagation and scattering in random media},
\newblock Wiley-IEEE Press, 1999.

\bibitem{Chatelain2013}
F.~Chatelain and N.~{Le Bihan},
\newblock ``von-mises fisher approximation of multiple scattering process on
  the hypershpere,''
\newblock in {\em Icassp}, 2013.

\bibitem{Lefebvre2006}
M.~Lefebvre,
\newblock {\em Applied stochastic processes},
\newblock Springer, 2006.

\bibitem{Saleh1978}
B.~Saleh,
\newblock {\em Photoelectron statistics},
\newblock Springer, 1978.

\bibitem{Ferrari2007}
A.~Ferrari, G.~Letac, and J.-Y. Tourneret,
\newblock ``Exponential families of mixed {P}oisson distributions,''
\newblock {\em Journal of Multivariate Analysis}, vol. 98, no. 6, pp.
  1283--1292, 2007.

\bibitem{Chatelain2009}
F.~Chatelain, S.~Lambert-Lacroix, and J.-Y. Tourneret,
\newblock ``Pairwise likelihood estimation for multivariate mixed poisson
  models generated by gamma intensities,''
\newblock {\em Statistics and Computing}, vol. 19, no. 3, pp. 283--301, Sept.
  2009.

\bibitem{Billingsley95}
Patrick Billingsley,
\newblock {\em {Probability and Measure}},
\newblock Wiley, New York, NY, 3rd edition, 1995.

\bibitem{karlis2005}
Dimitris Karlis and Evdokia Xekalaki,
\newblock ``Mixed {P}oisson distributions,''
\newblock {\em International Statistical Review}, vol. 73, no. 1, pp. 35--58,
  2005.

\bibitem{Purk98}
S.~Purkayastha,
\newblock ``Simple proofs of two results on convolutions of unimodal
  distributions,''
\newblock {\em Statistics \& Probability Letters}, vol. 39, no. 2, pp. 97--100,
  1998.

\end{thebibliography}

\begin{IEEEbiographynophoto}{Nicolas Le Bihan}
Nicolas Le Bihan obtained his B.Sc. degree in physics from the Université de
Bretagne Occidentale in Brest, France, in 1997. He received the
M.Sc. and Ph.D. degrees in signal processing in 1998 and 2001,
respectively, both from Grenoble INP. In 2011, he obtained the
Habilitation degree from Grenoble INP. Since 2002, he has been a
research associate at the Centre National de la Recherche
Scientifique (CNRS) and is working at the Department of Images
and Signals of the GIPSA-Lab (CNRS UMR 5083) in Grenoble,
France. From 2013 to 2015, he was a visiting fellow at the University of Melbourne, holding a Marie Curie International Outgoing Fellowship from the European Union (IOF GeoSToSip 326176), ERA, 7th PCRD.
His research interests include statistical signal processing
on groups, noncommutative algebras and differentiable manifolds
and its applications in polarized wave physics, waves in disordered
media, and geophysics.

\end{IEEEbiographynophoto}

\begin{IEEEbiographynophoto}{Florent Chatelain}
Florent Chatelain received the Eng. degree in
computer sciences and applied mathematics from
ENSIMAG, Grenoble, France, and the M.Sc. degree
in applied mathematics from the University Joseph
Fourier of Grenoble, France, both in June 2004,
and the Ph.D. degree in signal processing from the
National Polytechnic Institute, Toulouse, France, in
2007. He is currently an Assistant Professor at
GIPSA-Lab, University of Grenoble, France. His
research  interests are centered around estimation,
detection, and the analysis of stochastic processes. 
\end{IEEEbiographynophoto}

\begin{IEEEbiographynophoto}{Jonathan Manton}
Professor Jonathan Manton holds a Distinguished Chair at the University of Melbourne with the title Future Generation Professor. He is also an Adjunct Professor in the Mathematical Sciences Institute at the Australian National University. Prof Manton is a Fellow of IEEE and a Fellow of the Australian Mathematical Society.
He received his Bachelor of Science (mathematics) and Bachelor of Engineering (electrical) degrees in 1995 and his Ph.D. degree in 1998, all from the University of Melbourne, Australia. From 1998 to 2004, he was with the Department of Electrical and Electronic Engineering at the University of Melbourne. During that time, he held a Postdoctoral Research Fellowship then subsequently a Queen Elizabeth II Fellowship, both from the Australian Research Council. In 2005 he became a full Professor in the Department of Information Engineering, Research School of Information Sciences and Engineering (RSISE) at the Australian National University. From July 2006 till May 2008, he was on secondment to the Australian Research Council as Executive Director, Mathematics, Information and Communication Sciences. 
Prof Manton's traditional research interests range from pure mathematics (e.g. commutative algebra, algebraic geometry, differential geometry) to engineering (e.g. signal processing, wireless communications, systems theory). More recently, he has become interested in systems biology and systems neuroscience.
\end{IEEEbiographynophoto}

\end{document}